\documentclass[preprint,times,12pt]{elsarticle}

\usepackage[T1]{fontenc}

\usepackage{subcaption} 

\usepackage{booktabs} 

\usepackage{nicematrix} 

\usepackage[tworuled,algo2e]{algorithm2e} 

\usepackage{tikz}
\usetikzlibrary{shapes,arrows} 

\usepackage{cleveref}

\usepackage{url} 
\usepackage{geometry}
\usepackage{bm} 
\usepackage{amsthm} 
\usepackage{dsfont} 
\usepackage{amsmath} 
\DeclareMathAlphabet\mathbfcal{OMS}{cmsy}{b}{n}
\foreach \x in {a, ..., z}{
    \expandafter\xdef\csname bf\x \endcsname{\noexpand\ensuremath{\noexpand\mathbf{\x}}} 
}
\foreach \x in {A, ..., Z}{
    \expandafter\xdef\csname bf\x \endcsname{\noexpand\ensuremath{\noexpand\mathbf{\x}}} 
    \expandafter\xdef\csname ds\x \endcsname{\noexpand\ensuremath{\noexpand\mathds{\x}}} 
    \expandafter\xdef\csname cal\x \endcsname{\noexpand\ensuremath{\noexpand\mathcal{\x}}} 
    \expandafter\xdef\csname bcal\x \endcsname{\noexpand\ensuremath{\noexpand\mathbfcal{\x}}} 
}
\newcommand{\vUn}{\mathds{1}} 

\newcommand{\out}{\mathop{\otimes}} 
\newcommand{\T}{{\sf T}} 
\DeclareMathOperator*{\vectorize}{vec} 
\DeclareMathOperator*{\Card}{Card} 
\newcommand{\triples}[3]{\left\{#1,#2,#3\right\}} 
\newcommand{\jkl}{\triples{j}{k}{\ell}} 
\newcommand{\Am}[1]{\bfA^{(#1)}} 
\newcommand{\am}[2]{\bfa^{(#1)}_{#2}} 
\newcommand{\ua}{\underline{\bfa}} 
\newcommand{\uAm}[1]{\underline{\bfA}^{(#1)}} 
\newcommand{\uam}[2]{\ua^{(#1)}_{#2}} 
\newcommand{\uamr}[1]{\uam{#1}{r}} 
\newcommand{\amr}[1]{\am{#1}{r}} 
\newcommand{\lbd}{\bm{\lambda}} 
\newcommand{\cpdM}{\cpdsetp{\lbd; \Am{1},\ldots, \Am{M}}} 
\newcommand{\cpdsetp}[1]{\left[\!\left[#1\right]\!\right]} 
\newcommand{\Htrip}[3]{\bcalH^{(#1#2#3)}}
\newcommand{\Hjkl}{\Htrip{j}{k}\ell}
\def\cprank{{\mathrm{rank}_{\text{CP}}}} 
\def\cpprank{{\mathrm{rank}_{\text{CP+}}}} 
\newcommand{\model}{(\mu,\Tht)} 
\newcommand{\modelone}{(\mu_{1},\Tht_{1})} 
\newcommand{\modell}[1]{(\mu_{#1},\Tht^{#1}_1)} 
\newcommand{\tht}{\bm{\theta}} 
\newcommand{\Tht}{\bm{\Theta}} 
\newcommand{\jac}[1]{\bcalJ_{#1}}
\newcommand{\jacmt}{\bcalJ_{\mu}(\tht)}
\newcommand{\Rmax}{R_{\text{max}}}
\def\rank{{\mathrm{rank}}} 
\newcommand{\kron}{\mathop{\boxtimes}} 
\newcommand{\Bm}[1]{\bfB^{(#1)}} 
\newcommand{\Cm}[1]{\bfC^{(#1)}} 
\def\imag{{\mathrm{Im}}} 
\newtheorem{definition}{Definition}[section]
\newtheorem{lemma}[definition]{Lemma}
\newtheorem{remark}[definition]{Remark}
\newtheorem{proposition}[definition]{Proposition}
\newtheorem{theorem}[definition]{Theorem}
\newtheorem{example}[definition]{Example}
\newtheorem{conjecture}[definition]{Conjecture}

\journal{Signal Processing}

\begin{document}

\begin{frontmatter}

\title{Coupled tensor models for probability mass function estimation: Part II, Uniqueness of the model.}

\author{Philippe FLORES$^{a,1}$, Konstantin USEVICH$^{2}$, David BRIE$^{2}$}

\affiliation{organization={Corresponding author -- $^{1}$CNRS, Universit\'e Grenoble Alpes, Grenoble INP, GIPSA-lab ; 11 rue des Mathématiques, 38402 Saint-Martin-d'Hères, France -- $^{2}$CNRS, Universit\'e de Lorraine, CRAN ; Campus Sciences BP 70239, Vandœuvre-lès-Nancy, France}} 

\begin{abstract}
    In this paper, uniqueness properties of a coupled tensor model are studied.
    This new coupled tensor model is used in a new method called Partial Coupled Tensor Factorization of 3D marginals or PCTF3D.
    This method performs estimation of probability mass functions by coupling 3D marginals, seen as order-3 tensors.
    The core novelty of PCTF3D's approach (detailed in the part I article) relies on the partial coupling which consists on the choice of 3D marginals to be coupled.
    Tensor methods are ubiquitous in many applications of statistical learning, with their biggest advantage of having strong uniqueness properties.
    In this paper, the uniqueness properties of PCTF3D's constrained coupled low-rank model is assessed.
    While probabilistic constraints of the coupled model are handled properly, it is shown that uniqueness highly depends on the coupling used in PCTF3D.
    After proposing a Jacobian algorithm providing maximum recoverable rank, different coupling strategies presented in the Part I article are examined with respect to their uniqueness properties.
    Finally, an identifiability bound is given for a so-called \emph{Cartesian coupling} which permits enhancing sufficient bounds of the literature.
\end{abstract} 


\end{frontmatter}

\section{Introduction} \label{sec:intro}

While matrices can be decomposed with the Singular Value Decomposition \cite{eckart_approximation_1936}, high-order tensors \cite{comon_tensors_2014} can be factorized with similar decompositions such as the Tucker decomposition \cite{tucker_some_1966} or the Canonical Polyadic Decomposition (CPD) \cite{hitchcock_expression_1927}.
These decompositions are all connected to some notion of \emph{rank}.
For any of those decompositions, a question always remains: for which ranks can one ensure that a decomposition is unique ?
The main advantage of high-order tensors is that decompositions are unique under mild conditions -- that is for higher ranks -- compared to matrix decompositions \cite{kolda_tensor_2009}.
Also, tensor decompositions are unique only up to trivial ambiguities (scaling and permutations of rank-one terms).

Nevertheless, tensors suffer from the curse of dimensionality which states that the complexity of a problem increases exponentially with its number of dimensions.
In the context of probability mass function (PMF) estimation, this curse prevents to scale tensor methods to large datasets.
In a Part I article \cite{pctf3d_part1}, a PMF estimation method called Partial Coupled Tensor Factorization of 3D marginals (PCTF3D) was proposed.
PCTF3D tames the curse of dimensionality as it couples partially 3D-marginals -- seen as order-3 tensors -- to obtain the higher-order CPD of a PMF.
PCTF3D is a generalization of the method \cite{n_kargas_tensors_2018} (called FCTF3D for Full Coupled Tensor Factorization of 3D marginals) that permits to reduce the complexity while keeping reasonable estimation performances.

Although PCTF3D shows promising results, it lacks uniqueness guarantees compared to FCTF3D.
Indeed, PCTF3D introduced a new constrained coupled model.
In this article, the uniqueness properties of the model introduced by PCTF3D are studied.
To do so, both recoverability and identifiability are examined.

The aim of this paper is to study both recoverability and identifiability of our newly-proposed coupled tensor model.
To do so, Section \ref{sec:preli} presents preliminary results on the CP decomposition. In Section \ref{sec:cpd_simplex}, the probabilistic simplex constraints are introduced with some uniqueness results in this particular case.
In Section \ref{sec:PolAddModel}, we propose an algorithm that provides maximum recoverable rank in the case of polynomial additive models.
Then, Section \ref{sec:recoAlgo} shows that PCTF3D's model is, after reparameterization linked to simplex constraints, in fact a polynomial additive model.
Key results on PCTF3D's recoverability are then presented in Section \ref{sec:bounds}, while Section \ref{sec:balVSrngreco} gives insights on recoverability for the two mains coupling strategies presented in \cite{pctf3d_part1}: random and balanced.
Finally, a sufficient identifiability condition is proven in Section \ref{sec:cartesianProd}.

\subsection*{Notations}
Scalars (respectively vectors and matrices) are denoted as both upper and lowercase (respectively lowercase bold and uppercase bold) letters.
A higher-order tensor is an array whose number of dimensions $M\geq3$ and is denoted with a calligraphic bold letter.
The set of integers $\{1,\ldots, M\}$ is denoted $\cpdsetp{1,M}$ in the following.
The notation $\vUn_R$ refers to a column vector of size $I$ whose entries are all ones.
$\vectorize(.)$ denotes the column-major vectorization operator for a tensor or a matrix.
The operator $.^\T$ applied to a vector or a matrix designates the transposition.
The identity matrix of size $R$ will be denoted as $\bfI_R$.
Finally, the operator $\Card\left\{.\right\}$ denotes the number of elements of a given set.

\section{Background: the Canonical Polyadic tensor decomposition}\label{sec:preli}

\subsection{Definition of the CP decomposition} \label{sec:tens:defCPD}
An order-$M$ tensor of size $I_1\times\dots\times I_M$ is an $M$-way array with entries $\bcalX = [\calX_{i_1,\ldots,i_M}]_{i_1,\ldots,i_M=1}^{I_1,\ldots, I_M}$ (note that tensors are in fact multilinear operators \cite{comon_tensors_2014}).
Matrices, vectors and scalars can be viewed respectively as tensors of order 2, 1 and 0, but, in what follows, the term \emph{tensor} will refer to tensors of order $M\geq3$.
\begin{definition}[Outer product and rank-one tensors]
The outer (tensor) product of a set of $M$ vectors is an order-$M$ tensor $\bcalX$ of size $I_1\times\dots\times I_M$ whose entries are defined by:
\begin{equation*}
    x_{i_1\cdots i_M} = \prod\limits_{m=1}^M a^{(m)}_{i_m},
\end{equation*}
where $\left( \bfa^{(m)} \right)_{m=1}^M$, $\bfa^{(m)}\in\dsR^{I_m}$ is a tuple of vectors.
The outer product is usually denoted as $\bcalX = \bfa^{(1)}\out \cdots \out \bfa^{(M)}$.
 A tensor $\bcalX$ is said to be \emph{rank-one} if it can be written as an outer product.
\end{definition}

\begin{definition}[CPD, \cite{hitchcock_expression_1927}]
    The \emph{Canonical Polyadic Decomposition} (CP decomposition or CPD) of a tensor $\bcalX\in\dsR^{I_1\times\cdots\times I_M}$ is a decomposition of $\bcalX$ a sum of $R$ rank-one tensors:
    \begin{equation}
        \bcalX = \sum_{r=1}^R \lambda_r\amr{1}\out\cdots\out\amr{M},
        \label{eq:tens:fullcpd}
    \end{equation}
    where $\amr{m}\in\dsR^I$ for $m\in\cpdsetp{1,M}$ and $r\in\cpdsetp{1,R}$ are called the \emph{factors} of the decomposition and $\lambda_r \in \dsR$.
    For convenience, the factors of the decomposition \eqref{eq:tens:fullcpd} are grouped into \emph{factor matrices} $\Am{m}$, $m\in\cpdsetp{1,M}$ and the weights in $\lambda_r$ into a \emph{loading vector} $\lbd$:
    \begin{equation*}
    \Am{m} = \begin{bmatrix}\am{m}{1} & \cdots & \am{m}{R}\end{bmatrix}, \quad \lbd = \begin{bmatrix}\lambda_1 & \cdots & \lambda_R\end{bmatrix}^\T,
    \end{equation*}
    resulting in the compact CPD notation:
    \begin{equation}
        \bcalX = \cpdM.
        \label{eq:tens:cpd}
    \end{equation} 
\end{definition}

\begin{definition}[rank of a tensor]
    For a higher-order tensor $\bcalX$, the \emph{tensor rank} of $\bcalX$ (or also the \emph{CP rank of} $\bcalX$) is defined as the smallest integer such that \cref{eq:tens:fullcpd} holds and will be denoted $\cprank(\bcalX)$.
\end{definition}
There exist other notions of rank \cite{kolda_tensor_2009} (multilinear, block-term) which are different from the CP-rank.
However, in the rest of the paper, only the CPD will be used hence the term \emph{rank} will only refer to the CP-rank.

For non-negative tensors (tensors with non-negative entries), the \emph{non-negative rank} of $\bcalX$, denoted as $\cpprank(\bcalX) \geq \cprank(\bcalX) $ is the proper extension of the rank accounting for non-negative factors.
However, there are many cases where the two ranks coincide $\cpprank(\bcalX) = \cprank(\bcalX) $ \cite{qi_semialgebraic_2016}.
This is the situation considered in the remaining of the paper.

\subsection{Uniqueness of the CPD} \label{sec:tens:uniqueCP}

The biggest advantage of tensor decompositions over their matrix counterparts is their uniqueness under relatively mild conditions.

\begin{definition} 
The CPD \eqref{eq:tens:cpd} of a tensor $\bcalX$ is called unique is all other possible $R$-term CPDs are given by permutation and rescaling of elements, that is
\[
\bcalX = \cpdsetp{\bm{D}_0 \bm{\Pi}\lbd;\Am{1}\bm{\Pi}\bm{D}_1,\ldots,\Am{M}\bm{\Pi}\bm{D}_M},
\]
where $\bm{\Pi}$ is an $R\times R$ permutation matrix and $\bm{D}_m$ is a set of diagonal $R\times R$ nonsingular matrices, and $\bm{D}_0 = {\prod\limits_{m=1}^M \bm{D}^{-1}_m}$.

The nonnegative CPD of a nonnegative tensor is called unique if the same holds for nonnegative diagonal matrices.
\end{definition}

Denoting $ \kappa(\bfA)$, the Kruskal rank of matrix $\bfA$ -defined as the maximum number of columns such that every subset is linearly independent - the Kruskal condition gives a sufficient uniqueness condition for order-3 tensors \cite{kruskal_three-way_1977}.
\begin{proposition}[Kruskal condition, \cite{kruskal_three-way_1977}]
    If the matrices $\Am1\in\dsR^{I_1\times R}$, $\Am2\in\dsR^{I_2\times R}$ and $\Am3\in\dsR^{I_3\times R}$ satisfy
    \begin{equation*}
        \kappa(\Am{1})+\kappa(\Am{2})+\kappa(\Am{3}) \geq 2R+2,
    \end{equation*}
    then, the tensor $\bcalX = \cpdsetp{\Am1,\Am2,\Am3}$ is of rank $R$ and has a unique CP decomposition.
    \label{prop:tens:Kruskal3D}
\end{proposition}
A generalization of Kruskal's sufficient condition to order-$M$ tensors was provided in \cite{sidiropoulos_uniqueness_2000} and states that a CPD \eqref{eq:tens:cpd} is unique (for $\lbd$ without zeros) if:
\begin{equation}\label{eq:tens:kruskalM}
    2R+(M-1) \leq \sum\limits_{m=1}^M \kappa(\Am{m}).
\end{equation} 

\subsection{Generic uniqueness}\label{sec:chap2:generic}

In general, a property is said \emph{generic} if it holds almost everywhere \cite{comon_generic_2009}.
In this subsection, we introduce the notion of generic uniqueness in the case of tensor decompositions and some classical results around this notion.
\begin{definition} {\textbf{Generic uniqueness} -- }
    The CP model of rank $R$ \eqref{eq:tens:cpd} is said to be \emph{generically unique} (or \emph{identifiable}) if the CPD is unique for all possible factors except a set of factors of Lebesgue measure zero.
    Equivalently, the CPD is unique (with probability 1) for factor matrices drawn from an absolutely continuous distribution.
\end{definition}
In other words, if generic uniqueness is ensured, the probability of drawing random factors such that the CPD is not unique is zero.
\begin{example}
    By using the generalized Kruskal condition \eqref{eq:tens:kruskalM} in the case of cubic tensors, the CPD is generically unique if 
    \begin{equation*}
        2R+(M-1) \leq \min(I,R)M.
    \end{equation*}
    Indeed, it is known that $\kappa(\Am{m}) = \min(I,R)$ for factor matrices drawn from an absolutely continuous distribution.
    Therefore, the bound follows by replacing the Kruskal rank with its generic value in \eqref{eq:tens:kruskalM}.
\end{example}

In \cite{chiantini_generic_2012}, milder sufficient conditions were proved in the case of tensors of order 3, stating that the CPD is generically unique if
\begin{equation*}
    R\leq 4^{\log_2(I)-1}.
\end{equation*}
These are the results that were used in \cite{n_kargas_tensors_2018} for deriving uniqueness conditions for coupled factorization of marginals.

There are many recent results on generic rank and generic uniqueness and in many cases, the CP models for ranks below the generic ranks are identifiable \cite{chiantini_algorithm_2014,qi_semialgebraic_2016}.
In particular, the following result will be used.
\begin{theorem}[{\cite[Corollary 6.2]{bocci_refined_2014}}]
    Let $\bcalX$ a tensor of size $I_1\times I_2 \times I_3$ such that $2<I_3\leq I_2\leq I_1$. Then, the (complex-valued) $R$-term CP decomposition of $\bcalX$ is identifiable if:
    \begin{equation}
        R\leq \frac{I_1 I_2 I_3}{I_1+I_2+I_3-2}-I_1.
        \label{eq:thm_ident_3d}
    \end{equation}
\end{theorem}

Note that in the cubic case ($I_1=I_2=I_3=I$), the identifiability results go back to Strassen \cite{strassen_rank_1983}.
This condition shows that for cubic tensors, the maximum rank grows quadratically in $I^2/3$ which is better than the growing in $I^2/16$ of \cite{chiantini_generic_2012}.
Stronger results are available in \cite{chiantini_algorithm_2014} for small (not too large) $I$.
\begin{remark}
    Many identifiability results in the literature are proved for complex-valued tensors.
    However, as shown in \cite[Corollary 30]{qi_semialgebraic_2016}, for the same tensor size $I \times J \times K$, complex identifiability implies real identifiability and also nonnegative identifiability.
    \label{rem:tens:equivCRRplus}
\end{remark}

\section{Uniqueness of coupled CP decompositions with simplex constraints} \label{sec:cpd_simplex}

\subsection{Coupled CP decompositions with simplex constraints} 
\label{sec:pctf3drecalls}
Let $\bcalH$ be an (unknown) order-$M > 3$ tensor following a rank-$R$ CP model with simplex constraints 
\begin{equation}
 \bcalH = \cpdsetp{\lbd;\Am1,\ldots,\Am{M}}, 
    \end{equation}
subject to :\\
\begin{minipage}{0.45\linewidth}
    \begin{equation} \label{eq:lbdNN}
        \lbd\geq0
    \end{equation}
    \begin{equation} \label{eq:lbdS21}
        \vUn^\T_R\lbd = 1
    \end{equation}
\end{minipage} \hfill \begin{minipage}{0.45\linewidth}
    \begin{equation} \label{eq:AmNN}
        \Am{m}\geq0
    \end{equation}
    \begin{equation} \label{eq:AmS21}
        \vUn^\T_I\Am{m} = \vUn^\T_R
    \end{equation}
\end{minipage}
\\[2ex]
This is exactly the Probability Mass Function (PMF) model for high-dimensional discrete random vectors, described in \cite{pctf3d_part1}.
Defining a 3D-marginal as
\[
(\Hjkl)_{i_j i_k i_\ell} := 
\sum\limits^{I}_{\substack{i_t =1, \\ t \in \cpdsetp{1,M} \setminus \jkl}}  \mathcal{H}_{i_1\ldots i_M},
\]
i.e., the summation is performed along all modes excluding $\{j,k,l\}$, 
then each $\{\Hjkl\}_{\jkl\in\calT}$ admit a rank-$R$ CP decomposition with simplex constraints. Then, considering a collection of $T$ order-3 known tensors, indexed by triplets
\begin{equation}\label{eq:cpd_coupled}
    \Hjkl = \cpdsetp{\lbd;\Am{j},\Am{k},\Am{\ell}}, \quad \text{for all } \jkl\in\calT,
\end{equation}
\[
 \calT = \left\{\{j_1,k_1,\ell_1\},\{j_2,k_2,\ell_2\},\ldots,\{j_T,k_T,\ell_T\} \right\} \subset 2^{\cpdsetp{1,M}}
\] defines a rank $R$ coupled CP decompositions of order-3 tensors with simplex constraints. 
Following \cite[Section 4]{pctf3d_part1}, the fully coupled CP decompositions will refer to the case where all the $\binom{M}{3}$ possible triplets are considered while the so-called \emph{partial coupling} comes to the choice of $T< \binom{M}{3}$ triplets.

The question at hand is to give conditions on the maximal rank under which the coupled decompositions yield unique loading and factor matrices $\lbd;\Am1,\ldots,\Am{M}$ and thus uniquely recover the unknown high-order tensor $\bcalH$.

\subsection{Uniqueness (identifiability results)}
Two identifiability conditions are proved in \cite{n_kargas_tensors_2018}. 
\begin{theorem}[{\cite[Theorem 1]{n_kargas_tensors_2018}}]
\label{th:kargas_1}
    Let $\bcalH$ a tensor of size $I^M$ having a rank-$R$ CP model with simplex constraints.\\
    If $M \leq I$, then, $\bcalH$ is almost surely (a.s) identifiable from the full set of 3D-marginals $\Hjkl$ if 
    $$R \leq I(M-2)$$.
    If $M > I$, then, $\bcalH$ is a.s. identifiable from the full set of 3D-marginals $\Hjkl$ if 
    $$
    R \leq \left( \lfloor \frac{\sqrt{MI-1}}{I} \rfloor I - 1 \right)^2 
    $$
\end{theorem}
\begin{theorem}[{\cite[Theorem 2]{n_kargas_tensors_2018}}]
\label{th:kargas_2}
    Let $\bcalH$ a tensor of size $I^M$ having a rank-$R$ CP model with simplex constraints and let $\alpha$ be the largest integer such that $2 ^\alpha \leq \lfloor \frac{N}{3} \rfloor I $
    \\
    Then, $\bcalH$ is almost surely (a.s) identifiable from the full set of 3D-marginals $\Hjkl$ if
    $$R \leq 4^{\alpha - 1}$$
    which is implied by 
    $$R \leq \frac{\left(\lfloor \frac{M}{3} \rfloor I + 1\right)^2}{16}.$$
\end{theorem}
Depending on the values of $I$ and $M$, one can choose the most favorable rank bound.
However, these rank bounds are far from the maximum attainable rank for both the fully and partially coupled CP decompositions.
Examining in details the proofs of theorems \ref{th:kargas_1} and \ref{th:kargas_2}, it appears that they do not account for the simplex constraints of the coupled CP decompositions.
Also, they rely on a particular type of coupling that will be referred to as \emph{Cartesian coupling} which serves as a basis ingredient for proving the sufficient identifiability condition in section \Cref{sec:cartesianProd}.
To properly account for the simplex constraints, results from real algebraic geometry are adapted to the case considered.
Also, the relaxed notion of \emph{recoverability} of $\bcalH$ allows deriving rank bounds depending explicitly on the coupling strategy.

\section{Polynomial additive models and identifiability}\label{sec:PolAddModel}

In order to study the coupled uniqueness properties of our model, the Canonical Polyadic Decomposition is seen as a particular case of algebraic models: polynomial additive models.
In this subsection, we recall results on uniqueness of polynomial additive models from \cite{breiding_algebraic_2021} which uses the tools of real algebraic geometry \cite{qi_semialgebraic_2016}.
We will see that this framework permits handling simplex constraints properly for the study of both recoverability and identifiability. 

\subsection{Polynomial and additive mappings}
\begin{definition}
A polynomial model is a couple $(\mu, \Tht)$ of a polynomial map $\mu: \dsR^{n} \to \dsR^S$ and a set of parameters $\Tht \subseteq \dsR^{n}$, which is assumed to be polyhedral (i.e., a set defined by affine inequalities $\mathbf{A} \mathbf{x} \le \mathbf{b}$).
\end{definition}
\begin{remark}
Examples of polyhedral sets include $\Tht = \dsR^{n}$ (the whole space), $\Tht = \dsR^{n}_{+}$ (positive orthant) and $\Tht = \Delta_{n}$ (simplex).
Most of the statements in this section can be applied to the more general case when $\Tht$ is a semialgebraic subset \cite{qi_semialgebraic_2016} (defined by polynomial inequalities), but, for simplicity, only the polyhedral case will be considered. 
\end{remark}

\begin{definition}[Additive model] \label{def:additive}
Let us consider a base polynomial model $\modelone$ where $\Tht_1\subseteq\dsR^{n_1}$.
Then, the $R$-term additive model $\modell{R}$ is defined as 
\[
\mu_R(\tht) = \mu_1(\tht_1)+\cdots+\mu_1(\tht_R),
\]
where $\tht = \left(\tht_1, \cdots, \tht_R \right)$, $\tht_r \in \Tht_1$ for all $r =1,\ldots, R$.
\end{definition}

\begin{example} \label{ex:cpModel}
    The CP model can be viewed as a special case of polynomial additive models.
    In this case, we can choose the mapping $\mu_1: \dsR^{1+IM} \to \dsR^{I^{M}}$ that maps the factors $(\lambda_1, \am{1}{1},\ldots, \am{M}{1})$ to a (vectorized) rank-one tensor
    \begin{equation}
        \label{eq:paramRank1}
        \mu_1(\tht_r) = \vectorize\left(\lambda_1 \am{1}{1}\out\cdots\out\am{M}{1}\right),
    \end{equation}
So for $\tht = \left(\tht_1, \cdots, \tht_R \right)$, and $\tht_r = (\lambda_r, \am{1}{r},\ldots, \am{M}{r})$, 
we get
\[
\mu_R(\tht) = \sum\limits_{r=1}^R \vectorize\left( \lambda_r \amr{1}\out\cdots\out\amr{M}\right) = \vectorize\left(\cpdM\right),
\]
Choosing $\Tht_1 = \dsR^{1+IM}$ corresponds to the real CP model while choosing $\Tht = \dsR^{1+IM}_{+}$ corresponds to the non-negative CP model.
\end{example}

Additive polynomial models generalize tensor decompositions and can be studied in the framework of X-rank decompositions \cite{qi_semialgebraic_2016}.
Similarly to tensor decompositions, uniqueness of additive models can be defined.
\begin{definition}[Essential uniqueness of additive models]
    \label{def:essential}
Let $\modell{R}$ be an $R$-term additive model.
We say that $\modell{R}$ is \emph{essentially unique} at a given $\tht\in \left(\tht_1, \cdots, \tht_R \right)$ if for 
\begin{equation}\label{eq:unqueness_decomp}
\bfy = \mu_R(\tht) = \mu_1(\tht_1)+\cdots+\mu_1(\tht_R)
\end{equation}
all alternative decompositions arise only from permuting terms in \eqref{eq:unqueness_decomp} (formally, elements in $\mu^{-1} (\bfy)$ are equal to $\tht$ up to a permutation of terms $\tht_r$ and change of parameters $\tht_r$ that do not change the one-term result $\mu_1(\tht_r)$).

The model $\modell{R}$ is called \emph{identifiable} if it is essentially unique for a generic $\tht\in\Tht$ (that is for all $\tht$ except a subset of Lebesgue measure zero).
\end{definition}

For the case of CPD or nonnegative CPD (see \Cref{ex:cpModel}), uniqueness and identifiability in the sense of \Cref{def:essential} corresponds exactly to uniqueness and generic uniqueness of CP decompositions.

\subsection{Unique parameterization of additive models and non-ambiguous parameterizations}
In this paper, it will be more convenient to work with non-ambiguous parameterizations of the polynomial additive models,
so that the uniqueness in the sense of \Cref{def:essential} is equivalent to uniqueness of parameters except permutations.
\begin{definition}
For a model $\modelone$, we say that it is non-ambiguous at $\theta$ if $\mu_1^{-1}(\mu_1(\theta)) = \{\theta\}$ (equivalently, $\mu_1(\theta) = \mu_1(\theta')$ implies $\theta = \theta'$).
\end{definition}
Note that for tensor models, because of the scale ambiguity, the default parameterization \eqref{eq:paramRank1} is ambiguous, but it becomes non-ambiguous if we restrict the model as follows.

\begin{example}\label{def:rank1_parameterization}
For $I_1,I_2,I_3$, define a model $\mu_1: \Tht_1 \to \dsR^{I_1I_2I_3}$, with $\Tht_1 = \dsR^{I_1+I_2+I_3-2}$ which maps $\theta = (\lambda,\underline{\bfa},\underline{\bfb},\underline{\bfc})$ with $\lambda \in \dsR$, $\underline{\bfa} \in \dsR^{I_1-1}$,$\underline{\bfb}\in \dsR^{I_2-1}$,$\underline{\bfc} \in \dsR^{I_3-1}$,
\begin{equation}\label{eq:rank1_parameterization_ones}
\mu_1: \theta \mapsto\vectorize\left( \lambda \left[\begin{smallmatrix}\underline{\bfa} \\1 \end{smallmatrix}\right] \out \left[\begin{smallmatrix}\underline{\bfb} \\1 \end{smallmatrix}\right] \out \left[\begin{smallmatrix}\underline{\bfc} \\1 \end{smallmatrix}\right] \right).
\end{equation}
The parameterization in \eqref{eq:rank1_parameterization_ones} is non-ambiguous for almost all $\theta \in \Theta_1$ (except the ones with $\lambda=0$).
It also represents almost all rank-one $I_1 \times I_2 \times I_3$ tensors (i.e., except a set of measure zero).
\end{example}

\begin{remark}
    Parameterization \eqref{eq:rank1_parameterization_ones} is often used for studying statistical properties of tensor models, such as Cram\'{e}r-Rao bounds for CPD \cite{prevost_constrained_2022}.
\end{remark}

\begin{remark}[Uniqueness in non-ambiguous parameterizations]
    Consider the $R$-term polynomial additive model $\modell{R}$ and $\theta = (\theta_1,\ldots,\theta_R)$ so that $\mu_1$ is non-ambiguous at each $\theta_k$.

    Then $\modell{R}$ is unique at $\theta$ if and only if all other decompositions of $\bfy = \mu_R(\theta)$ in \eqref{eq:unqueness_decomp} are obtained by permuting $\theta_k$.
\end{remark}

In this paper, we will use non-ambiguous parameterizations for studying identifiability of coupled CP modes, and thus we will need the following lemma.
\begin{lemma}\label{lem:reduced_parameterization}
If the $R$-term CPD of an $I_1 \times I_2 \times I_3$ tensor is generically unique, then the $R$-term additive polynomial model for \eqref{eq:rank1_parameterization_ones} is identifiable
\end{lemma}
\begin{proof}
The proof of this technical lemma is contained in \ref{sec:tech_proofs}.
\end{proof}

\subsection{Recoverability and Jacobian of parameterization}
A key tool for studying identifiability is the relaxed notion of recoverability \cite{breiding_algebraic_2021}.
\begin{definition}[Recoverability] \label{def:recoverability}
    The model $(\mu,\Tht)$ is said \emph{recoverable} at $\tht\in\Tht$ if there exists a finite number of elements in the pre-image of $\mu^{-1}(\mu(\tht))$.
    The model $(\mu,\Tht)$ is \emph{generically recoverable} if it is recoverable for a generic $\tht\in\Tht$ (i.e., for all $\tht$ except a subset of Lebesgue measure zero).
\end{definition}
For polynomial additive models, recoverability is a necessary condition to identifiability.
We formulate this fact for the case of non-ambiguous parameterization, although it can be formulated for more general cases in \cite{breiding_algebraic_2021}.
\begin{remark}[Identifiability implies recoverability]\label{rem:identIrecov}
    Let $\modell{R}$ be an identifiable $R$-term polynomial additive model, such as $\mu_1$ is a non-ambiguous parameterization for almost all $\theta_1 \in \Tht_1$.
    Then $\modell{R}$ is generically recoverable.
\end{remark}

The notion of recoverability is particularly convenient because it can be studied using the Jacobian of the parameterization.
\begin{definition}[Jacobian matrix of a parametrization]
    For a model $(\mu,\Tht)$, the \emph{Jacobian of the parametrization} at a given $\tht = \begin{bmatrix}
        \theta_1 & \cdots & \theta_n
    \end{bmatrix}^\T \in\Tht$ is an $S\times n$ matrix defined by:
    \begin{equation*}
        \jacmt = \left( \frac{\partial y_s}{\partial \theta_i} \right)_{s=1,i=1}^{S,n},
    \end{equation*}
    where $\bfy$ is the vector such that $\bfy = \mu(\tht) = \begin{bmatrix}
        y_1(\tht) & \cdots & y_S(\tht)
    \end{bmatrix}^\T$.
\end{definition}
Then the following theorem links generic recoverability to the rank of the Jacobian.
\begin{proposition}[{Special case of \cite[Th. 4.9]{breiding_algebraic_2021}}]\label{prop:q5brei}
    The following statements are equivalent:
    \begin{enumerate}
        \item There exists a $\tht^*\in\mathbb{R}^{n}$ such that the $ \jac{\mu}(\tht^{*})$ is full column rank.
        \item The model $(\mu,\mathbb{R}^{n})$ is generically recoverable.
    \end{enumerate}
\end{proposition}
\begin{remark} \label{rem:certificate}
    In the previous proposition, if at a single point $\tht^*$ the Jacobian is full column rank, then it is full column rank for a generic point $\tht\in\mathbb{R}^{n}$ (\emph{almost everywhere} on $\mathbb{R}^{n}$).
    This follows from \cite[Lemma 1]{qi_semialgebraic_2016}.
\end{remark}

\subsection{An algorithm for checking recoverability of polynomial additive models}
\Cref{prop:q5brei} is particularly useful to check recoverability of additive polynomial models, as we show below. 
Indeed, if the rank of Jacobian is maximal at a point $\tht\in\Tht_1^R$, i.e. 
\begin{equation}
    \rank(\jacmt) = Rn_1,
    \label{eq:fullRank}
\end{equation}
then the model $\modell{R}$ is recoverable by \Cref{prop:q5brei}.
Note that the additivity of the model ensures that the Jacobian matrix has the following block structure:
\begin{equation}
    \jacmt = \begin{bmatrix}
        \jac{\mu_1}(\tht_1) & \cdots & \jac{\mu_1}(\tht_R)
    \end{bmatrix} \label{eq:blockRank}
\end{equation}
where $\jac{\mu_1}(\tht_r)$ represents the Jacobian of the rank-1 parametrization $\mu_1$ applied to the $r$-th block of parameters $\tht_r$.
If \Cref{eq:fullRank} holds, then any subsets of block columns of the Jacobian is also full column rank.
This means that the same $R'$-term models $R'< R$ are also generically recoverable.
Let us denote $\Rmax$ the integer such that \eqref{eq:fullRank} holds for $R\leq \Rmax$ and does not hold for $R > \Rmax$.
We call such $\Rmax$ the recoverability bound because the model $(\mu,\mathbb{R}^{n})$ is necessary not recoverable if the Jacobian is not full column rank.

The above discussion leads us to the following algorithm for checking the identifiability numerically (which is a generalization of the approach used in \cite{comon_generic_2009} for tensors).
\begin{algorithm2e}[hbt!]
        \renewcommand{\algorithmcfname}{Algorithm}
        \SetAlgoLined
        \KwIn{$M$, $I$, $\calT$}
        \KwSty{Initialization}{: $R=1$, $\tht = \tht_1 \in \dsR^{n_1}$ random.}
    
        \While{$\rank(\jacmt) = Rn_1$}{
    
        $\tht_{R+1} \in \dsR^{n_1}$ random. \\
    
            $\tht = (\tht_1,\ldots,\tht_{R},\tht_{R+1})$ \\
            
            $R \leftarrow R+1$
          } 
        \KwOut{$\widehat{\Rmax} = R-1$}
        \caption{Lower bound on maximum recoverable rank $\Rmax$ for an additive polynomial model $\modell{R}$.}
        \label{alg:rmax}
\end{algorithm2e}

The principle of \Cref{alg:rmax} is to increase the rank of a randomly generated decomposition while \eqref{eq:fullRank} holds for the model defined in \Cref{def:additive}.

\begin{remark}
    In \Cref{alg:rmax}, a parameter $\tht$ that satisfies \eqref{eq:fullRank} gives a certificate of generic recoverability for $R \le \widehat{\Rmax}$.
    This follows from Remark \ref{rem:certificate}.
    Such a point $\tht$ presents a computer proof of recoverability, as long as the absence of numerical errors can be guaranteed (this can be done, for example, by choosing parameters with rational entries, as in \cite{chiantini_algorithm_2014}).
\end{remark}
Note that it may happen that a chosen random parameter in $\tht$ the Jacobian drops rank, this is why $\widehat{\Rmax}$ returned by \Cref{alg:rmax} gives a lower bound.
However, we can try several realizations of $\tht$ (or $\tht_{R+1}$), to overcome this problem.

\subsection{Recoverability and identifiability for submodels and reparameterization}
We close this section by several remarks on submodels and equivalent models.

\begin{lemma}\label{lem:identifiable-submodel}
    Let $\modell{R}$ an identifiable (respectively generically recoverable) additive model and consider a subset of parameters $\widetilde{\Tht}_1 \subseteq \Tht_1$ of same dimension (i.e., ${\widetilde{\Tht}_1}$ is not Lebesgue measure zero in ${\Tht}_1$), then $(\mu_r,\widetilde{\Tht}_1^R)$ is identifiable (respectively recoverable).
\end{lemma}
\begin{proof}
    We provide a proof for the case of identifiability.
    Let $\calY$ be the subset of $\Tht$ that contains the non-identifiable parameters.
    Because $\model$ is identifiable, $\calY$ is of Lebesgue measure zero in $\Tht$.
    Then, $\calY\cap\widetilde{\Tht}$ is also of Lebesgue measure zero in $\widetilde{\Tht}$ which implies identifiability of $(\mu,\widetilde{\Tht})$.
\end{proof}

\begin{remark}
    \Cref{lem:identifiable-submodel}, in fact, explains why identifiability of real CPD implies identifiability of nonnegative CPD.
    Indeed, this follows from the fact that the positive orthant $\widetilde{\Tht}_1 = \dsR^{n}_{+}$ is a polyhedral subset of the same dimension of $\Tht_1 = \dsR^{n_1}$.
    This argument will be particularly useful for studying coupled tensor models.
\end{remark}

Finally, we note that a change of parameterization does not change the identifiability/recoverability of the model.

\begin{remark}\label{rem:linear_reparameterization}
Let $(\mu,{\Tht}_1)$ be an additive polynomial model and $(\widetilde{\mu},\widetilde{\Tht}_1)$ be a linear reparameterization.
(i.e., let $\phi: \widetilde{\Tht}_1 \to \Tht$ be an affine bijective map, where $\widetilde{\Tht}_1 \subset \dsR^{\widetilde{n}_1}$ and $\widetilde{\mu} = \mu_1\circ \phi$.

Then the model defined by $(\widetilde{\mu}, \widetilde{\Tht}_1)$ is identifiable (resp. recoverable) if and only if $({\mu}, {\Tht}_1)$ identifiable (resp. recoverable).
\end{remark}

\section{Coupled tensor models as polynomial additive mappings} \label{sec:recoAlgo}

\subsection{PCTF3D as a polynomial map: unconstrained case}
In this subsection, we are going to represent the PCTF3D model as a polynomial map, mapping the weights $\lbd$ and the factors $\Am{1},\ldots,\Am{M}$ to the marginals $\Hjkl$.
We first consider the case without constraints (allowing $\Am{m} \in \dsR^{I \times R}$ and $\lbd \in \dsR^{I}$).
Note that, by definition of the CPD, \eqref{eq:cpd_coupled} can be equivalently rewritten as 
\begin{equation}
 \Hjkl = \sum\limits_{r=1}^R \lambda_{r} \amr{j} \circ \amr{k} \circ \amr{\ell}, \quad \text{for all } \jkl\in\calT \label{eq:cpd_coupled_additive}
\end{equation}
In case $R=1$ (rank-one model), the mapping from factors to coupled rank-one marginals \eqref{eq:cpd_coupled_additive} is the following polynomial map:
\begin{equation}
    \label{eq:model_rank1}
    \begin{array}{cc|ccc}
        \mu_1 & : & \dsR^{IM+1} & \to & \dsR^{TI^3} \\
            & & \tht_1 = \left(\lambda_1, \am{1}{1}, \ldots, \am{M}{1}\right) & \mapsto &\lambda_{1} \vectorize \Big( \am{j_1}{1} \circ \am{k_1}{1} \circ \am{\ell_1}{1},\ldots,\am{j_M}{1} \circ \am{k_M}{1} \circ \am{\ell_M}{1} \Big) \\
        \end{array},
\end{equation}
where each size-$(I^3)$ subvector of $\mu(\tht)$ corresponds to a 3D marginal $\Htrip{j}{k}{\ell}$, $\jkl \in \calT$.

Then the rank-$R$ case can (without constraints) can be considered as an additive model, which follows from linearity of \eqref{eq:cpd_coupled_additive}.
Consider splitting of the parameter vector as follows:
\begin{equation} \label{eq:model_parameters_blocks}
\tht = (\tht_1,\ldots,\tht_R)\in \dsR^{R(IM+1)}, \text{ with } 
\tht_r = \left(\lambda_r, \am{1}{r}, \ldots, \am{M}{r}\right) \in \dsR^{IM+1}.
\end{equation}
Then the following lemma holds true
\begin{lemma}
The tensors $\{\Htrip{j}{k}{\ell} \}_{(j,k,\ell) \in \calT}$ are obtained as \eqref{eq:cpd_coupled_additive} (or \eqref{eq:cpd_coupled}) from the factors $\Am{1},\ldots,\Am{M}$ and $\lbd$
if and only if
\begin{equation}\label{eq:pctf3d_model_unconstrained}
\Big( \vectorize {\Htrip{j_1}{k_1}{\ell_1},\ldots,\vectorize \Htrip{j_T}{k_T}{\ell_T} } \Big) = \mu(\tht) := \mu_1(\tht_1) + \cdots + \mu_1(\tht_R),
\end{equation}
with $\mu_1$ as in \eqref{eq:model_rank1} and $\tht_r$ is as in \eqref{eq:model_parameters_blocks}.
\end{lemma}

\subsection{Taking into account and relaxing the sum to one constraint}
The map $\mu$ defined in \eqref{eq:pctf3d_model_unconstrained} is an additive model in the sense of \Cref{def:additive} if we don't take into account constraints on factors.
Now, if we add constraints (\ref{eq:lbdNN}-\ref{eq:AmS21}) in the model \eqref{eq:pctf3d_model_unconstrained}, then they are represented by the set of parameters
\begin{equation}\label{eq:Tht_pctf3d}
    \Tht = \left\{\left(\lambda_1, \am{1}{1}, \ldots, \am{M}{1}, \ldots, \lambda_R, \am{1}{R}, \ldots, \am{M}{R} \right) \in \dsR^{R(IM+1)}: \lbd\geq0, \vUn^\T_R\lbd = 1, \am{m}{r}\geq0, \vUn^\T_I\am{m}{r} = 1
    \right\}.
\end{equation}
However, the set $\Tht$ in \eqref{eq:Tht_pctf3d} cannot be written as $\Tht_1^{R}$, so the constrained (original PCTF3D) model is not strictly an additive model in the sense of \Cref{def:additive}, due to the constraint $\lambda_1 + \cdots + \lambda_R = 1$.
However, the constraint set remains invariant under permutations of blocks of parameters $\tht_r$, and identifiability of the PCTF3D model can be introduced as follows.
For this we formally introduce identifiability of the PCTF3D model as follows.
\begin{definition}\label{def:PCTF3D_uniqueness}
    For fixed $I,M$, $\calT$, and $R$, let $(\mu,\Tht)$ be the PCTF3D model defined by \eqref{eq:pctf3d_model_unconstrained} with $\Tht$ as in \eqref{eq:Tht_pctf3d}.
    Such model $(\mu,\Tht)$ is called identifiable if for a general $\tht \in \Tht$, $\mu^{-1}(\tht)$ consists only of block permutations of $\tht_k$ in \eqref{eq:model_parameters_blocks}.
\end{definition}
In the rest of the subsection, we will show that the constraint set can be enlarged (relaxed) to obtain an additive model in the sense of \Cref{def:additive}, and we can study identifiability/recoverability for the relaxed constraint set instead PCTF3D. 

Consider the enlarged set of constraints
\begin{equation} \label{eq:Tht_pctf3d_ext}
    \Tht' = \left\{\left(\lambda_1, \am{1}{1}, \ldots, \am{M}{1}, \ldots, \lambda_R, \am{1}{R}, \ldots, \am{M}{R} \right) \in \dsR^{R(IM+1)}: \lambda_r\geq0, \am{m}{r}\geq0, \vUn^\T_I\am{m}{r} = 1 \right\}, \\
\end{equation}
so that $\Tht \subset \Tht'$.
It is easy to see that $\Tht' = \Tht_1^{R}$ with 
 \begin{equation}\label{eq:Tht_pctf3d_ext_rank1}
    \Tht_1 = \left\{\left(\lambda_1, \am{1}{1}, \ldots, \am{M}{1} \right) \in \dsR^{IM+1}: \lambda_1\geq0, \am{m}{1}\geq0, \vUn^\T_I\am{m}{1} = 1 \right\} = \dsR_+ \times \Delta_{I}^{M}.
\end{equation}

We conclude the subsection by the desired relaxation proposition, which will allow us to work with the constraint set \eqref{eq:Tht_pctf3d_ext} from now on.

\begin{proposition}\label{prop:relaxing_sum_to_one}
    The partially coupled model (i.e., $(\mu, \Tht)$ with $\Tht$ as in \eqref{eq:Tht_pctf3d}) is generically unique in the sense of \Cref{def:PCTF3D_uniqueness} (resp. generically recoverable in the sense of \Cref{def:recoverability}) if and only if the model $(\mu, \Tht')$ (i.e., the polynomial additive model $(\mu,\Tht') = \modell{R}$ with $\mu_1$ as in \eqref{eq:model_rank1} and $\Tht_1$ as in \eqref{eq:Tht_pctf3d_ext_rank1}) is identifiable (resp. generically recoverable).
\end{proposition}
\begin{proof}
    The proof is contained in \ref{sec:tech_proofs}.
\end{proof}

\subsection{Reparameterization of the simplex constraint on factors} \label{sec:reparametrization}
The final building block of our proof is to reparameterize the simplex constraint by removing variables, which essentially follows the idea of \Cref{def:rank1_parameterization}.
For that, consider the parametrization $\calP$ which maps a vector $\ua\in\mathbb{R}^{I-1}$ onto a sum-to-one vector $\bfa = \calP(\ua)$:
\begin{equation}
    \label{eq:projTrunc}
    \calP(\ua) = \begin{bmatrix}
        a_1 &
        \cdots &
        a_{I-1} &
        1-\sum\limits_{i=1}^{I-1} a_i
    \end{bmatrix}^\T.
\end{equation}
With truncated factors $\uamr{m} = \begin{bmatrix}
a^{(m)}_{r,1} \cdots a^{(m)}_{r,I-1} \end{bmatrix}^{\T}$, it is possible to define a vector of parameters $\tht$ such that:
\begin{equation}
    \label{eq:deftht}
    \tht = (\tht_1,\ldots,\tht_R) = \vectorize\left(\lambda_1, \uam{1}{1}, \ldots, \uam{M}{1}, \ldots, \lambda_R, \uam{1}{R}, \ldots, \uam{M}{R} \right),
\end{equation}
and a set of parameters $\Tht_1$ defined by
\begin{equation}
    \label{eq:defTht}
    \Tht_1 = \left\{ \tht\in\dsR^{1+M(I-1)} \;|\; \lbd\geq0, \; \amr{m}\geq0, \; \vUn^\T_I\amr{m}\leq1 \right\} = \dsR_+ \times \blacktriangle_{I-1}^{M},
\end{equation}
where $\blacktriangle_{I-1}$ denotes the simplex interior. 
The advantage of $\Tht_1$ defined in \eqref{eq:defTht} is that it is a subset of positive Lebesgue measure of $\dsR^{1+M(I-1)}$.

Then it is possible to define the function $\mu$ that maps the parameters onto the set of 3D marginals of $\calT$:
\begin{equation}
    \mu :\left| \begin{array}{ccl}
        \mathbb{R}^{(1+M(I-1))R} & \longrightarrow & \mathbb{R}^{TI^3} \\
        \tht & \longmapsto & \bfy = \vectorize\left( \bcalH^{(\tau_1)}, \ldots, \bcalH^{(\tau_T)} \right) \\
    \end{array} \right., \label{eq:defMu} 
\end{equation}
where 
\begin{equation}
    \bcalH^{(\tau_t)} = \cpdsetp{\lbd;\calP(\uAm{j_\tau}),\calP(\uAm{k_\tau}),\calP(\uAm{\ell_\tau})},
    \label{eq:mappinghtau}
\end{equation}
and the reduced factor matrices ${\uAm{m}}\in\mathbb{R}^{(I-1)\times R}$ are defined as $\uAm{m} = \begin{bmatrix}
    \uam{m}{1} & \cdots & \uam{m}{R} 
\end{bmatrix}$.

\begin{remark}\label{rem:reparameterization_PCTF3D}
The model $(\mu,\Tht)$ defined by \eqref{eq:defTht}-\eqref{eq:mappinghtau} is a polynomial additive model as introduced in \Cref{sec:PolAddModel}.
In addition, by \Cref{rem:linear_reparameterization}, this model is equivalent to the polynomial additive model $\modell{R}$ with $\mu_1$ as in \eqref{eq:model_rank1} and $\Tht_1$ as in \eqref{eq:Tht_pctf3d_ext_rank1}).
\end{remark}

\section{Recoverability: key results}\label{sec:bounds}
\subsection{Reminder: hypergraphs, triplets and pairs}
Let us recall the basic considerations on the couplings $\calT$.
Recall the following notation.
\begin{itemize}
\item $T =\Card(\calT)$ the number of triplets;
\item $d_m$ denotes the number of triplets in which the variable (vertex) is contained; i.e., $d_m = \Card\left\{ \triples{m}k\ell \in\calT\right\}$; in particular, we have
\begin{itemize}
\item $d_1 + \cdots + d_M = 3T$, as each triplet is counted three times; 
\item $d_m \ge 1$ for couplings considered in this work (from \cite{pctf3d_part1}, a valid coupling must contain all $M$ variables);
the vector $\bfd = (d_1,\ldots,d_M)$ is called the sequence degree vector.
\end{itemize} 
\item We will also need the number of pairs appearing in triplets:
\[
P = \Card\left\{ \{j,k\}\in\cpdsetp{1,M} \;|\; j\neq k \; \text{and} \; \exists \tau\in\calT, \{j,k\}\subset\tau \right\},
\]
i.e., the number of edges in the usual graph (2-graph) induced by the hypergraph.
\end{itemize} 
These quantities will influence identifiability bounds as we will show below.
We first remark that $d_m$, $P$ and $T$ are linked and not all combinations are possible.
To give some figures, distribution of both quantities are given in \Cref{fig:repart} for different values of $T$ in the case of $M=7$ variables.
\begin{figure} 
    \centering 
    \begin{subfigure}[b]{0.45\linewidth}
    \centerline{\includegraphics[height=8cm, clip, trim={50 20 70 60}]{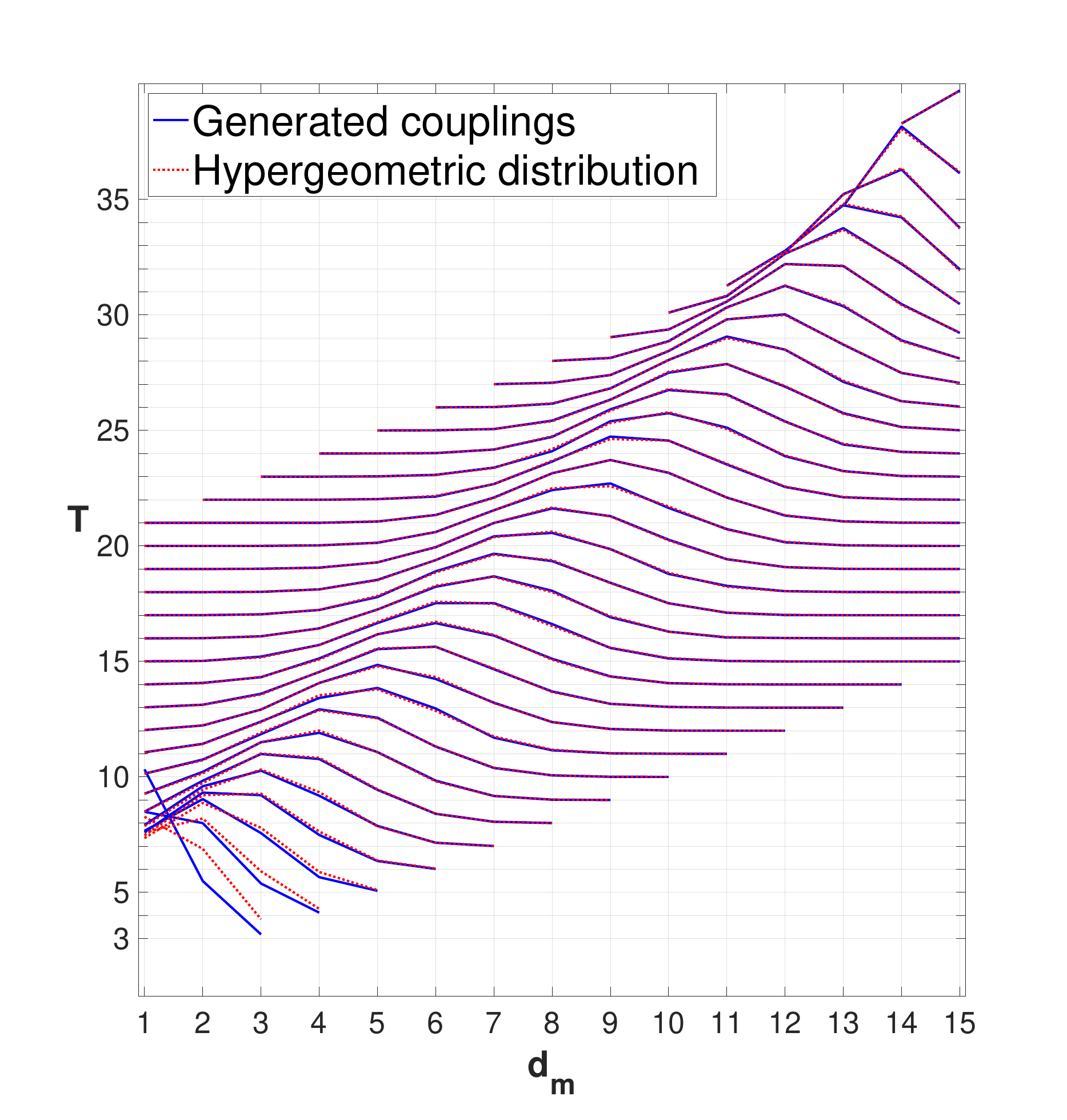}}
    \caption{ \centering Both empirical and theoretical distribution of $d_m$.}
    \label{fig:repart:deg}
    \end{subfigure}
    \hfill
    \begin{subfigure}[b]{0.45\linewidth} 
    \centerline{\includegraphics[height = 8cm, clip, trim={50 20 70 60}]{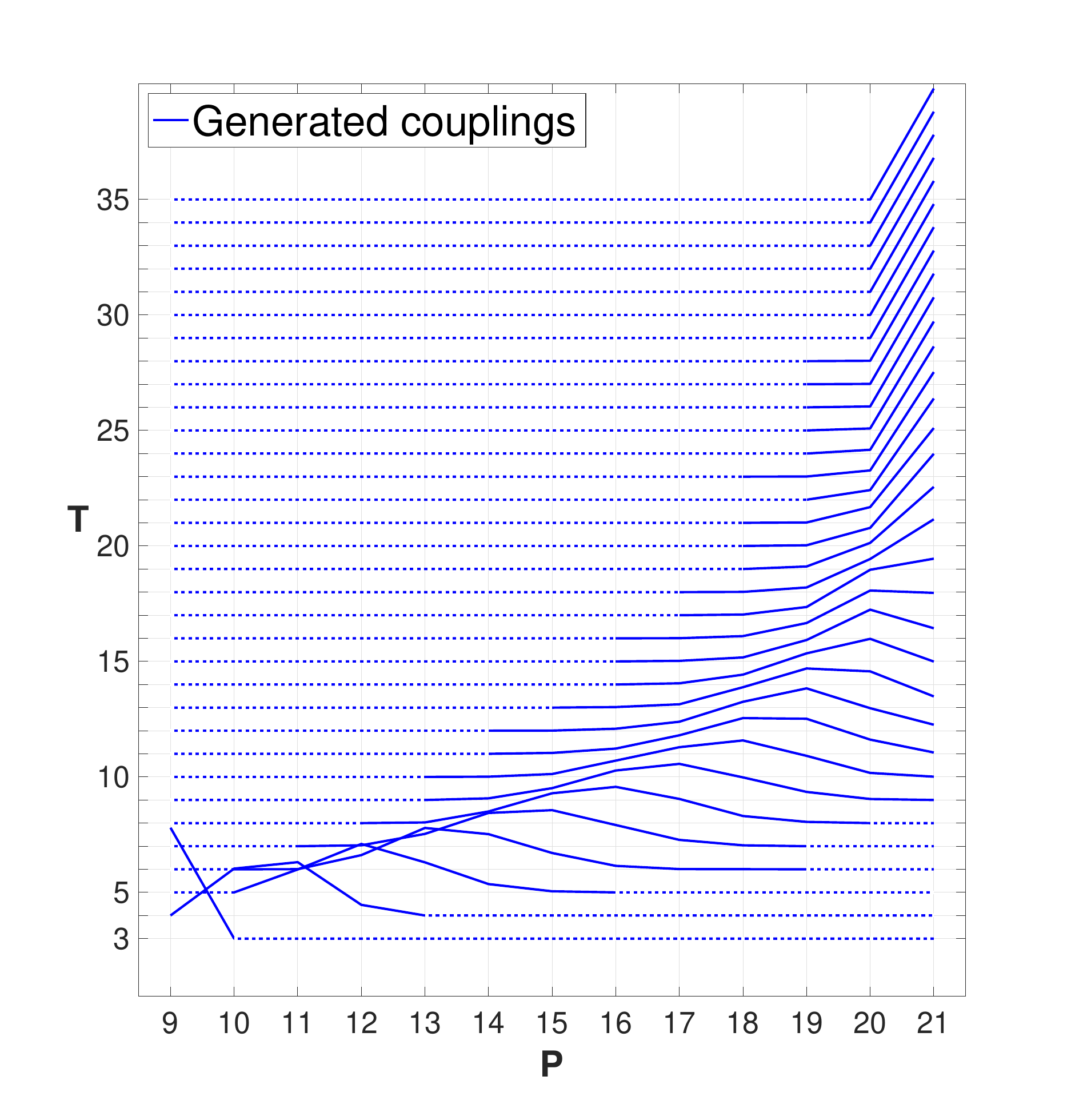}}
    \caption{\centering Empirical distribution of $P$.}
    \label{fig:repart:pair}
    \end{subfigure}
    \label{fig:repart}
    \caption{Distributions of both $d_m$ and $P$ for random couplings and different values of $T$ ($M=7$).}
\end{figure}

Without validity constraints on $\calT$ (ensured by the connectedness of $\calT$), $d_m$ follows a hypergeometric distribution, as seen in \Cref{fig:repart:deg}.
Concerning $P$, \Cref{fig:repart:pair} shows that $P$ increases with $T$, but achieves its maximum possible number $\binom{M}2$ starting around $T = 25$.

\subsection{Recoverability necessary condition and number of degrees of freedom} \label{sec:degreeFreedom}
Our main result shows that the maximum attainable recoverable rank for the PCTF3D model is governed by the following constant.
\begin{proposition}
        \label{prop:necessaryCondT}
        For a coupling strategy $\calT$, if a model $(\mu,\mathbb{R}^{n})$ with $\mu$ given in \Cref{eq:defMu} is recoverable, then
        \begin{equation}
            R \leq \left\lfloor \frac{ 1+M(I-1)+P(I-1)^2+T(I-1)^3}{1+M(I-1)}\right\rfloor.
            \label{eq:neceCond}
        \end{equation}
\end{proposition}

The proof of proposition follows from the results in \Cref{sec:PolAddModel}.
In particular, by \Cref{prop:q5brei}, if a CP model is recoverable, then the Jacobian matrix $\jacmt$ must be full column rank.
Thus, a necessary condition for recoverability is that the number of columns of $\jac{\mu}(\tht)$ does not exceed the dimension of the image space.

In the following proposition, the dimension of the image is obtained by counting the degrees of freedom of the space of observations.
This number is actually smaller than $TI^3$ (the number of elements in $\mu(\tht)$), due to redundancy of information in different marginals.
\begin{lemma}
        \label{prop:nobs}
        For a coupling $\calT$, the dimension of the image of the parametrization $\mu$ (from \eqref{eq:defMu}) denoted $N_{\text{obs}}(\calT)$ and is given by:
        \begin{equation}
                N_{\text{obs}}(\calT) = \dim(
\imag(\mu)) = 1+M(I-1)+P(I-1)^2+T(I-1)^3.
                \label{eq:NobsT}
        \end{equation}
\end{lemma}
\begin{proof}
    To count the number of possibly different observations of the set of $T$ marginals, we separate the count for each lower-order marginals.
    In our case of 3D-marginals coupling, we must count up to order-3 marginals.
    The lower-order marginals can be obtained from the observations by contracting the 3D marginals with the vector of ones.
    There is $1$ order-0 marginal which is the sum of all entries of $\bcalY$ ($\bfy = \mu(\tht)$).
    Because, we relaxed the sum-to-one constraint on $\lbd$, this sum is not fixed hence count as 1 in the sum \eqref{eq:NobsT}.

    For 1D-marginals, there are $M$ different marginals.
    For the marginal $\bfh^{(m)}$, the sum over its entries is equal to the order-0 marginal.
    Therefore, the $I$ values of $\bfh^{(m)}$ are constrained by one linear equation.
    Hence, each 1D-marginal contributes to $I-1$ values.

    With the same reasoning, an order-2 marginal contributes for $(I-1)^2$ free entries as well as an order-3 marginal contributes for $(I-1)^3$ new free entries.
    The proof is complete because there are $T$ order-3 marginals and $P$ different pairs of variables hence $P$ order-2 marginals.
\end{proof}

\Cref{prop:nobs} is key to proving proposition \Cref{prop:necessaryCondT}. 

\begin{proof}[Proof of \Cref{prop:necessaryCondT}]
By \Cref{prop:q5brei}, if the model is recoverable, then the Jacobian matrix must be full column rank.
A necessary condition for that is (see also \eqref{eq:fullRank})
    \begin{equation*}
        R(1+M(I-1))\leq N_{\text{obs}}(\calT),
    \end{equation*}
with $ N_{\text{obs}}(\calT)$ as in \Cref{prop:nobs}, which leads to the proposed condition on $R$.
\end{proof}

\subsection{Recoverability results for fully coupled tensor factorization}

When all triplets are considered in $\calT$, PCTF3D becomes equivalent with full coupling hence FCTF3D which was presented in \cite{n_kargas_tensors_2018}.
The number of triplets in this case if $\binom{M}3$ and the number of pairs is $\binom{M}2$.
In this case, \Cref{prop:necessaryCondT} implies that the model can be recoverable only if:
\begin{equation}
    \label{eq:neceCondFull}
    R\leq \left\lfloor \frac{1+M(I-1)+\binom{M}2(I-1)^2+\binom{M}3 (I-1)^3}{1+M(I-1)} \right\rfloor.
\end{equation}

By applying the \Cref{alg:rmax} to this case, it is possible to determine the recoverability bound for the fully coupled case.
In \Cref{fig:RmaxFull}, the blue curve represents the most favorable identifiability sufficient condition proposed in \cite{n_kargas_tensors_2018}.
The dotted orange curve represents \Cref{eq:neceCondFull} which is a necessary condition for identifiability and recoverability.
Finally, the orange full line represents the rank $R_\text{max}$ obtained with \Cref{alg:rmax}.
\Cref{fig:RmaxFull} shows that the necessary identifiability condition is way above the sufficient conditions provided by \cite{n_kargas_tensors_2018}.
Moreover, recoverability results of \Cref{alg:rmax} are achieved for ranks up to the necessary condition of \Cref{eq:neceCondFull}.
\begin{figure}
        \centering
        \centerline{\includegraphics[width = \linewidth,clip,trim = {150 0 160 0}]{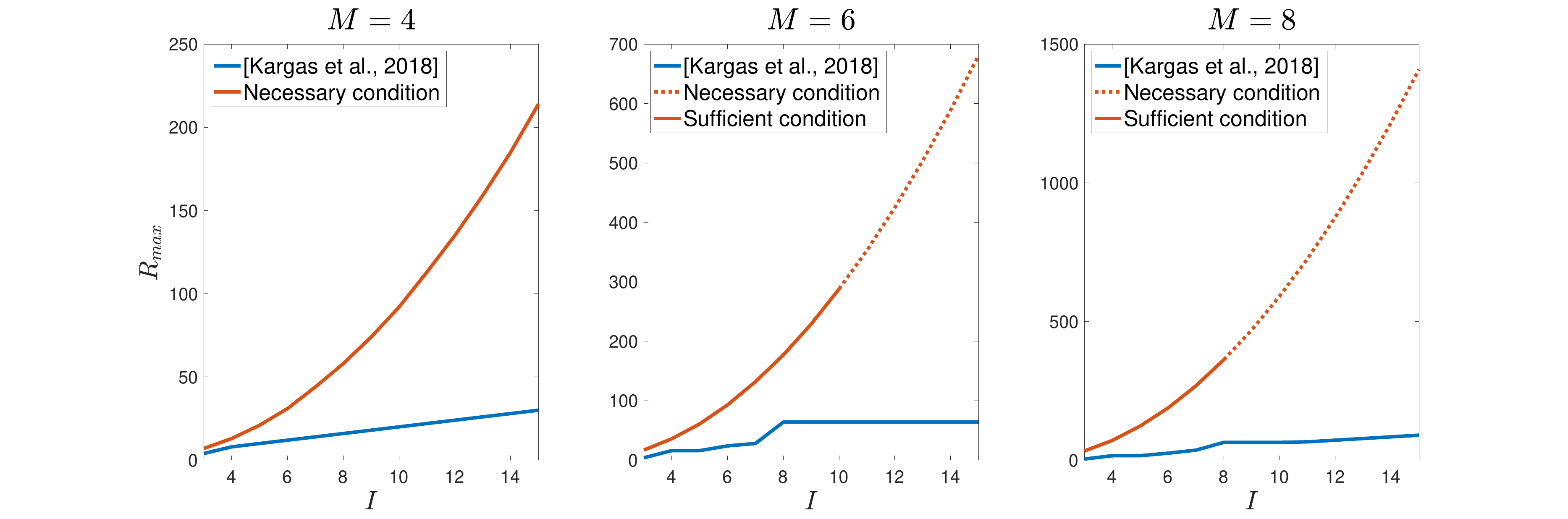}}
        \caption{Recoverability bounds for different values of $M$ and $I$ in the case of the FCTF3D.}
        \label{fig:RmaxFull}
\end{figure}

Our computations show that the model recoverability is achieved for all ranks up to the necessary condition of \Cref{prop:necessaryCondT}, which allows us to formulate the following conjecture.
\begin{conjecture}
    \label{conj:recovFull}
    For the full coupling $\calT$, the coupled CP model 
    \eqref{eq:pctf3d_model_unconstrained} is recoverable if and only if \eqref{eq:neceCondFull} holds.
\end{conjecture}

\section{Random and balanced couplings and analysis of defective cases} \label{sec:balVSrngreco}

In \cite{pctf3d_part1}, two main strategies were compared: random and balanced couplings.
While random couplings consist in picking $T$ triplets randomly, the balanced coupling strategy constrains the sequence of degrees to be as constant as possible.
By doing so, a balanced coupling ensures that any variable appears evenly in terms of occurrences in $\calT$.
In this subsection, we will compare recoverability for both coupling strategies.
To do such comparison, we generated couplings for both strategies and run \Cref{alg:rmax} for both strategies.

\paragraph{Random couplings}
To illustrate the behavior of the identifiability bound in the random case, we computed $\Rmax$ for 1000 different realizations of $\calT$ in the case of $M=8$ variables and $I=4$ bins per dimension.
\Cref{fig:RmaxRand} shows how the rank $R_\text{max}$ is distributed with respect to the values of $T$.
Unlike the fully coupled case which is deterministic, the bound \eqref{eq:neceCond} is not always achieved as the left plot of Figure \ref{fig:RmaxRand} suggests.
As the number of triplets increases, the bound is increasing which is coherent.
However, for a notable number of random couplings, $R_\text{max}$ is not exceeding the values 10 and 16.
Those two cases are going to be properly studied in the following subsection.
Consequently, the more triplets are considered, the less variability is expected in the recoverability bound.
\begin{figure}
        \centering
        \centerline{\includegraphics[width = \linewidth,clip,trim = {150 0 160 30}]{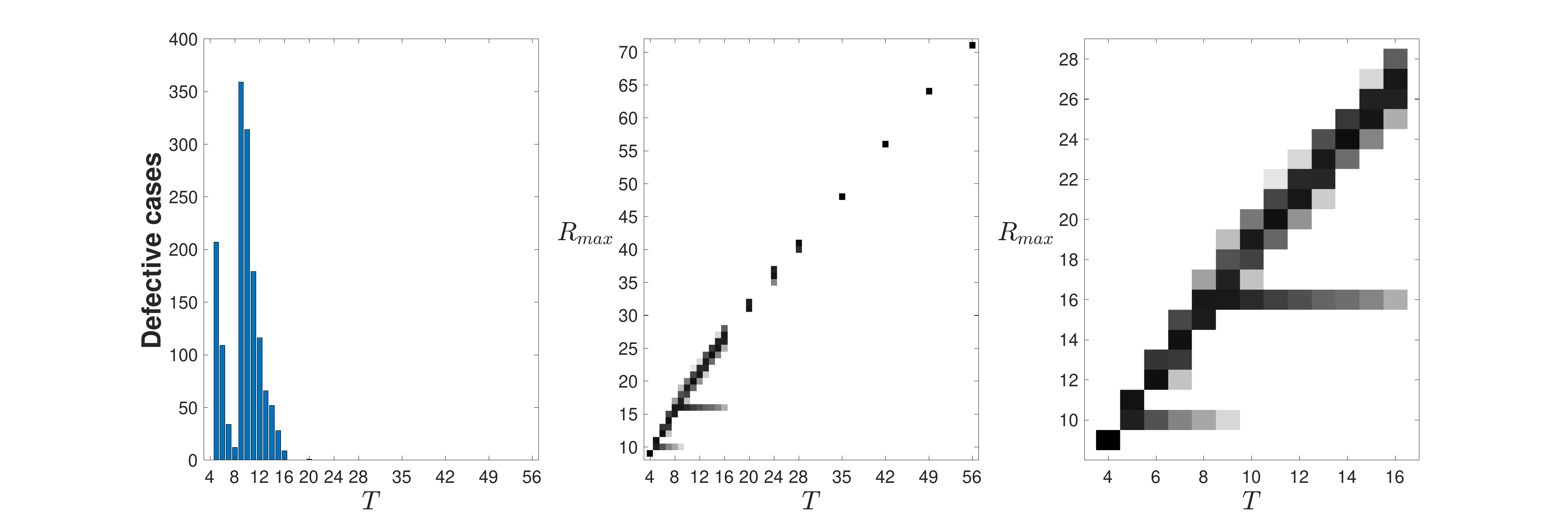}}
        \caption{Recoverability bounds in the case of random partial coupling.
        \textbf{Left plot}: Number of defectives cases (over 1000 realizations) for which $R_\text{max}$ is not equal to the maximal rank provided by the condition \eqref{eq:neceCond}.
        \textbf{Middle plot}: Distribution of $R_\text{max}$ regarding $T$.
        \textbf{Right plot}: Zoom on the middle plot around the lower values of $T$.
        Defective cases appear distinctly at ranks $R_\text{max}=10$ and $R_\text{max}=16$.}
        \label{fig:RmaxRand}
\end{figure} 

\paragraph{Balanced couplings}
Similarly to random couplings, we proceed to the same experiment where $M=8$ and $I=4$.
However, because the number of Lyndon words is limited, it is not possible to create 1000 different balanced couplings at a fixed value of $T$ (see \Cref{tab:numTrialsBal}).
\begin{table}
    \label{tab:numTrialsBal} 
    \centering
    \begin{tabular}{cccccccccccccc}
        \toprule
        $T$ & 4 & 5 & 6 & 7 & 8 & 9 & 10 & 11 & 12 & 13 & 14 & 15 & 16 \\ \midrule
        Number of trials & 36 & 10 & 60 & 30 & 187 & 30 & 13 & 247 & 13 & 7 & 195 & 21 & 313 \\
        \bottomrule
    \end{tabular}
    \caption{\centering Number of trials regarding the number of triplets considered in $\calT$ for the balanced couplings experiment.}
\end{table}
\begin{figure}
    \centering
    \centerline{\includegraphics[width = 0.6\linewidth,clip,trim = {50 10 110 10}]{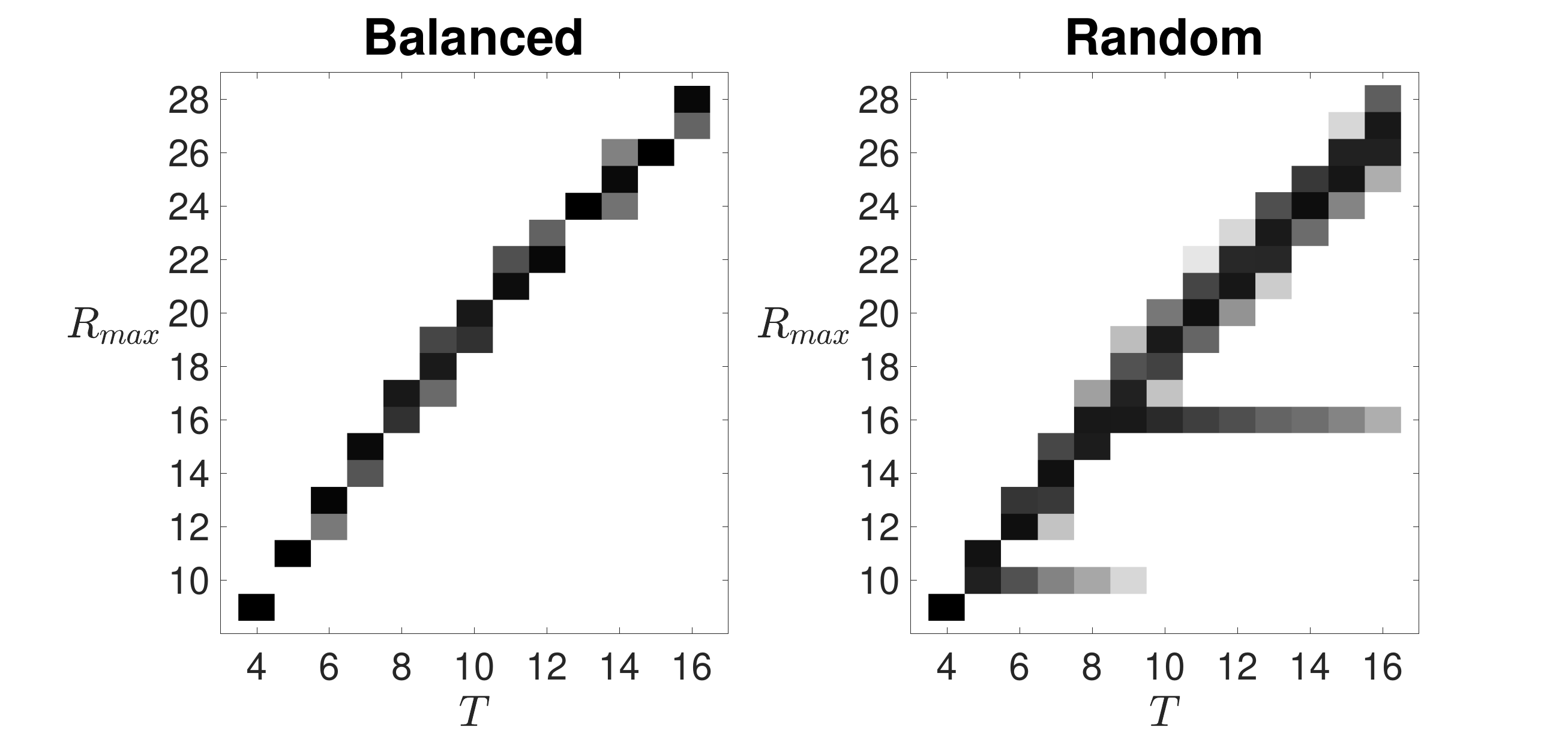}}
    \caption{Recoverability bounds in the case of balanced partial couplings.
    \textbf{Left plot}: Distribution of $R_\text{max}$ regarding $T$ for balanced couplings.
    \textbf{Right plot}: For comparison, the same distribution is shown for randomly chosen triplets.}
    \label{fig:RmaxBal}
\end{figure} 
\Cref{fig:RmaxBal} presents the empirical distribution of $R_\text{max}$ versus $T$.
For balanced strategies, the number of defect cases is not displayed because the recoverability is verified up to the necessary condition \eqref{eq:neceCond} in every case of the experiment.
We conjecture that for balanced couplings the defective cases never occur.
In the following subsections, defective cases observed in \Cref{fig:RmaxRand} are examined.

\subsection{Incidence matrix and Jacobian structure} \label{sec:jacobianMat}
In order to understand the peculiarities of random couplings, we need a closer look at the structure of the Jacobian.

\paragraph{Block structure of the Jacobian}
In \Cref{sec:PolAddModel}, we saw that the Jacobian matrix for an additive model can be separated column-wise in $R$ blocks \eqref{eq:blockRank}.
Because the model parametrized by \eqref{eq:defMu} is coupled, $\jacmt$ can also be separated in $T$ blocks.
Each block represents an observed marginal:
    \begin{equation}
        \jacmt = \begin{bmatrix}
        \jac{\mu}^{(\tau_1)}(\tht) \\
        \vdots \\
        \jac{\mu}^{(\tau_T)}(\tht)
    \end{bmatrix}
\end{equation}
where $\jac{\mu}^{(\tau_t)}(\tht)$ denotes the Jacobian of the parametrization of $\bcalH^{(\tau_t)}$ defined in \Cref{eq:mappinghtau}.

Recall from \cite{pctf3d_part1}, that the coupling $\calT$  as a connected 3-uniform hypergraph denoted.
The standard notions from hypergraph theory will be helpful to
\begin{definition}[Incidence matrix]
For a hypergraph $\calT$, is incidence matrix is a matrix of zeros and ones $\bfV\in\{0,1\}^{T\times M}$, where $T =\Card(\calT)$, where for the $t$-triplet $\jkl\in\calT$, the row $\bfv_{t,:}$ of $\bfV$ contains ones only in positions in $\jkl$.
\end{definition}

The incidence matrix is also helpful to represent various quantities related to hypergraphs.
Indeed, we have
\begin{enumerate}
\item The sequence of degree $\bfd$ of a coupling $\calT$ is a sum of columns of $\bfV$:
\begin{equation*}
    d_m = \Card\left\{ \triples{m}k\ell \in\calT\right\} = \vUn^T_T\bfV.
\end{equation*}
\item The number of pairs can be found as the number of non-orthogonal rows of $\bfV$: 
\begin{equation}
    P = \frac12 \left(\left\| \bfV^\T\bfV\right\|_0-M\right),
\end{equation}
where $\|\cdot\|_0$ denotes the number of non-zero elements in a matrix.
\end{enumerate}

Furthermore, the incidence matrix is also related to the structure of the Jacobian.
To give an example, consider $M = 4$ and the coupling strategy $\calT = \{\triples134,\triples234\}$.
Then the Jacobian matrix of the $r$-th rank one term (see \eqref{eq:blockRank}) is defined by 
\begin{equation}
    \label{eq:jac4Dexemple}
    \bcalJ_{\mu}(\tht_r) = \begin{bNiceMatrix}
        \amr{4}\kron\amr{3}\kron\amr{1} & & \Block{2-3}<\LARGE>{\lambda_r \bfB_r} & & \\ 
        \amr{4}\kron\amr{3}\kron\amr{2} & & & & &
    \end{bNiceMatrix},
\end{equation}
where
\begin{equation*}
    \bfB_r = \begin{bNiceMatrix}
        \amr{4}\kron\amr{3}\kron \bcalJ_{\calP} & \mathbf{0}_{I^3\times I} & \amr{4}\kron \bcalJ_{\calP}\kron\amr{1} & \bcalJ_{\calP}\kron\amr{3}\kron\amr{1} \\ 
        \mathbf{0}_{I^3\times I} & \amr{4}\kron\amr{3}\kron \bcalJ_{\calP} & \amr{4}\kron \bcalJ_{\calP}\kron\amr{2} & \bcalJ_{\calP}\kron\amr{3}\kron\amr{2}
    \end{bNiceMatrix}
\end{equation*}
and where $\jac{\calP}$ denotes the Jacobian matrix of the projection $\calP$ \eqref{eq:projTrunc}
\begin{equation}
    \label{eq:projS21}
    \jac{\calP} = \begin{bmatrix}
        1 & & \\
        & \ddots & \\
        & & 1 \\
        -1 & \cdots & -1
    \end{bmatrix}.
\end{equation}
Note that the structure of $\bfB_r$ in \eqref{eq:jac4Dexemple} is directly linked to the incidence matrix $\bfV$ of $\calT$.
\begin{equation*}
    \bfV = \begin{bmatrix}
        1 & 0 & 1 & 1 \\ 0 & 1 & 1 & 1
    \end{bmatrix}
\end{equation*}
This means that, excluding $\lbd$, each row block of $\jac{\mu}^{(\tau_t)}(\tht_r)$ only contains 3 non-zeros blocks out of $M$.

\subsection{Analysis of defective cases} \label{sec:defectCases}
In this subsection, we analyze the defective cases appearing in \Cref{fig:RmaxRand} and \Cref{fig:RmaxBal}.

In the following, let the sequence of degrees $\bfd$ be sorted in ascending order, which can always be the case up to a permutation of variables.
We present in this section two cases of a particular sequence of degrees that leads to identifiability and recoverability loss, which can explain the defective cases from the previous subsection.

\subsubsection*{First case: $d_1=1$ and $d_2>1$}
This setup means that the first variable (and only this one) appears once in $\calT$.
This case can occur as soon as $T\leq \binom{M-1}{2}$ but is more and more likely as $T$ is decreasing.
After analyzing the couplings of Figure \ref{fig:RmaxRand}, those cases are leading to low recoverability bounds equal to $\Rmax=16$.
\begin{proposition}
    Let $\calT$ a coupling strategy such that $d_1=1$ and $d_2>1$.
    If the coupled CP model is identifiable, then we have that 
    \begin{equation}
            \Rmax\leq I^2,
            \label{eq:defectcase1}
    \end{equation}
    for this coupling strategy.
    \label{prop:defectcase1}
\end{proposition}
\begin{proof}
    Let a coupled model with such coupling $\calT$ and suppose this model is identifiable.
    Then, it is necessary that $\jacmt$ is full rank.
    As seen in \Cref{sec:jacobianMat}, the structure of $\jacmt$ is linked to the incidence matrix $\bfV$:
    \begin{equation*}
            \bfV = \begin{bmatrix}
                    1 & v_{12} & \cdots & v_{1M} \\
                    0 & v_{22} & \cdots & v_{2M} \\
                    \vdots & \vdots & & \vdots \\
                    0 & v_{T2} & \cdots & v_{TM} 
            \end{bmatrix}.
    \end{equation*}
    Therefore, it is possible to define $\bfB_r\in \mathbb{R}^{I^3\times(I-1)}$ and $\bfC_r\in \mathbb{R}^{(T-1)I^3\times (M-1)(I-1)+1}$ such that:
    \begin{equation}
            \jac{\mu}(\tht_r) = \left[\begin{array}{c|ccc}
                \bfB_r & & \cdots & \\ \hline
                & & & \\
                \boldsymbol{0} & & \bfC_r & \\
                & & & \\
                \end{array} \right].
            \label{eq:structJacDefect}
    \end{equation}
    It follows that a necessary condition for $\jacmt$ is full rank is that the matrix $\bfB = \begin{bmatrix}
            \bfB_1 & \cdots & \bfB_R
    \end{bmatrix}$ is full rank.
    Because of sum-to-one constraints and thus the structure of $\jac{\calP}$, $\bfB$ contains $I^2$ dependent rows.
    Therefore, necessarily the number of observations $I^3-I^2$ should not be less than the number of parameters $R(I-1)$ which completes the proof.
\end{proof}

To verify this proposition, we proceed to the following experiment.
For $M\in\cpdsetp{4,15}$, Algorithm \ref{alg:rmax} was applied to the case where all triplets with variables $\left\{2,\ldots,M\right\}$ are present and only the triplet $\triples{1}{M-1}{M}$ contains the first variable.
This represents the most favorable case where one variable is present once in the coupling.
The results of this experiment are plotted in \Cref{tab:defectcase1}.
Even for this most favorable case, recoverability is guaranteed up to $I^2$ except for small number of $M$ where the condition \eqref{eq:neceCond} is more restrictive.
\begin{table}
        \label{tab:defectcase1} 
        \centering
        \begin{tabular}{cccccccc}
            \toprule
            M & & 4 & 5 & 6 & 7 & $\cdots$ & 15 \\ \midrule
            $I=3$ & & 5 & 7 & 9 & 9 & $\cdots$ & 9 \\ \midrule
            $I=4$ & & 8 & 13 & 16 & 16 & $\cdots$ & 16 \\ \midrule
            $I=6$ & & 18 & 32 & 36 & 36 & $\cdots$ & 36 \\ \bottomrule
        \end{tabular}
        \caption{\centering Evolution over $M$ of $R_\text{max}$ for most favorable coupling case with $d_1=1$.}
\end{table}

\subsubsection*{Second case: $d_1=d_2=1$ and $\tau_1=\triples{1}{2}{\ell_1}$}

In this case, the first two variables are present once in the first triplet and only in this triplet.
Note that $d_1=d_2=1$ ensures that $d_3>1$ because the hypergraph $\calT$ must be connected.
Indeed, if $d_1 = d_2 = d_3 = 1$ and the first triplet is $\tau_1 = \triples123$, then variables $\triples123$ are not connected to the other variables.
After analyzing the results of \Cref{fig:RmaxRand}, the couplings with $d_1=d_2=1$ lead to a recoverability bound of $\Rmax = 10$.
\begin{proposition}
    Let $\calT$ a coupling strategy such that $d_1=d_2=1$ and $\tau_1=\triples{1}{2}{\ell_1}$.
    If the coupled CP model is identifiable, then we have
    \begin{equation}
            \Rmax\leq \frac{I(I+1)}{2}.
            \label{eq:defectcase2}
    \end{equation}
    \label{prop:defectcase2}
\end{proposition}
\begin{proof}
    With the same method used in \Cref{prop:defectcase1}, let us consider an identifiable coupled model such that $d_1=d_2=1$ and $\tau_1=\triples{1}{2}{\ell_1}$.
    The incidence matrix $\bfV$ is equal to
    \begin{equation}
            \bfV = \begin{bmatrix}
                    1 & 1 & v_{13} & \cdots & v_{1M} \\
                    0 & 0 & v_{23} & \cdots & v_{2M} \\
                    \vdots & \vdots & \vdots & & \vdots \\
                    0 & 0 & v_{T3} & \cdots & v_{TM} \\
            \end{bmatrix},
    \end{equation} and therefore we can define $\bfB_r \in\mathbb{R}^{I^3\times 2(I-1)}$ and $\bfC_r \in\mathbb{R}^{(T-1)I^3\times (M-2)(I-1)+1}$ such that $\jacmt$ has the structure \eqref{eq:structJacDefect}.
    Because now 2 factors are included, only $I$ rows of $\bfB = \begin{bmatrix}
            \bfB_1 & \cdots & \bfB_R
    \end{bmatrix}$ are redundant.
    Hence, we obtain that $2R(I-1)\leq I^3-I$ which leads to the condition \eqref{eq:defectcase2}.
\end{proof}

By conducting a similar experiment, Table \ref{tab:defectcase2} shows that the bound verifies the condition \eqref{eq:defectcase2}.
In this case, all triplets with variables $\left\{3,\ldots,M\right\}$ were present and only the triplet $\triples12M$ contains the first two variables.
Maximal recoverability ranks are also bounded by $\frac{I(I+1)}{2}$ and \eqref{eq:neceCond} as expected.
\begin{table}
        \centering
        \begin{tabular}{cccccccc}
            \toprule
            M & & 5 & 6 & 7 & 8 & $\cdots$ & 15 \\ \midrule
            $I=3$ & & 4 & 6 & 6 & 6 & $\cdots$ & 6 \\ \midrule
            $I=4$ & & 7 & 10 & 10 & 10 & $\cdots$ & 10 \\ \midrule
            $I=6$ & & 16 & 21 & 21 & 21 & $\cdots$ & 21 \\ \bottomrule
        \end{tabular}
        \caption{\centering Evolution over $M$ of $R_\text{max}$ for most favorable coupling case with $d_1=d_2=1$.}
        \label{tab:defectcase2} 
\end{table}

To conclude this subsection, we would like to emphasize the fact that this list of defective cases may not be exhaustive.
Indeed, for higher number of variables, additional defective cases may appear, but they may be more difficult to characterize.

\section{Identifiability of Cartesian product coupling} \label{sec:cartesianProd}

\subsection{Cartesian coupling: definition}
In this section, we analyze special case of Cartesian coupling, for which we can prove stronger results on identifiability.
The following coupling strategy is called \emph{cartesian}:
\begin{equation}\label{eq:cartesian_coupling}
    \calT = \calM_1\times\calM_2\times\calM_3,
\end{equation}
where $\calM_1$, $\calM_2$, $\calM_3$ are three disjoint subsets of variables given by
\[
\calM_1 = \{j_1,\ldots,j_{M_1}\}, \quad
\calM_2 = \{k_1,\ldots,k_{M_2}\}, \quad
\calM_3 = \{\ell_1,\ldots,\ell_{M_3}\}.
\]
These couplings served as a base for proofs of \cite{n_kargas_tensors_2018}, however, here we provide a comprehensive analysis of the constraints, which helps us to obtain stronger results.

\begin{remark}[degrees and number of pairs]
    The cartesian coupling has the following $\bfd$ and $P$:
    \begin{equation*}
        d_m = \left\{ 
        \begin{array}{ll}
            M_2M_3 & \text{ if $m\in\calM_1$} \\
            M_1M_3 & \text{ if $m\in\calM_2$} \\
            M_1M_2 & \text{ if $m\in\calM_3$}
        \end{array}\right., \quad \text{ and } \quad P = M_2M_3+M_1M_3+M_1M_2,
    \end{equation*}
    where $M_1 = \Card(\calM_1)$, $M_2 = \Card(\calM_2)$ and $M_3 = \Card(\calM_3)$.
\end{remark}

\subsection{Stacking marginals into a 3D tensor}
The key idea is that under Cartesian coupling all data (marginals) can be stacked in to a $3$-way tensor $\bcalY \in \dsR^{IM_1 \times IM_2 \times IM_3}$.
Formally, we define $\bcalY$ to be a tensor composed of $M_1M_2M_3$ blocks indexed by $(r,s,t) \in \cpdsetp{1,M_1} \times\cpdsetp{1,M_2}\times\cpdsetp{1,M_3}$ such that, in MATLAB-like notation,
\[
\bcalY_{1+(r-1)I:rI,1+(s-1)I:sI,,1+(t-1)I:tI} = \Htrip{j_r}{k_s}{\ell_t}.
\]
The resulting tensor is visualized in \Cref{fig:repartKargas}
\begin{figure}
    \centering
    \resizebox{.7\linewidth}{!}{\tikzset{every picture/.style={line width=0.75pt}} 

\begin{tikzpicture}[x=0.75pt,y=0.75pt,yscale=-1,xscale=1]

\coordinate (Co111) at (110,360);
\coordinate (Co11M) at (450, 20);
\coordinate (Co1MM) at (950, 20);
\coordinate (CoMMM) at (950, 350);
\coordinate (CoMM1) at (610, 690);
\coordinate (CoM11) at (110, 690);
\coordinate (Co1M1) at (610, 360);

\draw  [line width=3.75]  (Co111) -- (Co11M) -- (Co1MM) -- (CoMMM) -- (CoMM1) -- (110,690) -- cycle ; \draw [line width=3.75]  (Co1MM) -- (Co1M1) -- (Co111) ; \draw [line width=3.75] (Co1M1) -- (CoMM1) ;

\draw (85,360) .. controls (80,360) and (78,364) .. (78,368) -- (78,515) .. controls (78,522) and (76,525) .. (71,525) .. controls (76,525) and (78,528) .. (78,535)(78,532) -- (78,682) .. controls (78,686) and (80,690) .. (85,690) ;
\draw (30,512) node [anchor=north west][inner sep=0.75pt]  [font=\LARGE] [align=left] {$M_{1}$};

\draw (110,715) .. controls (110,720) and (114,722) .. (118,722) -- (350,722) .. controls (357,722) and (360,724) .. (360,729) .. controls (360,724) and (363,722) .. (370,722)(367,722) -- (602,722) .. controls (606,722) and (610,720) .. (610,715) ;
\draw (342,742) node [anchor=north west][inner sep=0.75pt]  [font=\LARGE] [align=left] {$M_{2}$};

\draw    (624,700) .. controls (627,703) and (630,703) .. (633,700) -- (787,543) .. controls (792,538) and (796,537) .. (799,540) .. controls (796,537) and (797,533) .. (801,528)(799,530) -- (955,371) .. controls (958,368) and (958,364) .. (955,361) ;
\draw (802,542) node [anchor=north west][inner sep=0.75pt]  [font=\LARGE] [align=left] {$M_{3}$};

\draw   [line width=1.5]  (110,359.53) -- (160.54,308.99) -- (290,308.99) -- (290,426.92) -- (239.46,477.46) -- (110,477.46) -- cycle ; \draw   [line width=1.5]  (290,308.99) -- (239.46,359.53) -- (110,359.53) ; \draw   [line width=1.5]  (239.46,359.53) -- (239.46,477.46) ;
\draw   [line width=1.5]  (110,569.53) -- (160.54,518.99) -- (290,518.99) -- (290,636.92) -- (239.46,687.46) -- (110,687.46) -- cycle ; \draw   [line width=1.5]  (290,518.99) -- (239.46,569.53) -- (110,569.53) ; \draw   [line width=1.5]  (239.46,569.53) -- (239.46,687.46) ;
\draw   [line width=1.5]  (480,569.53) -- (530.54,518.99) -- (660,518.99) -- (660,636.92) -- (609.46,687.46) -- (480,687.46) -- cycle ; \draw   [line width=1.5]  (660,518.99) -- (609.46,569.53) -- (480,569.53) ; \draw   [line width=1.5]  (609.46,569.53) -- (609.46,687.46) ;
\draw   [line width=1.5]  (480,359.53) -- (530.54,308.99) -- (660,308.99) -- (660,426.92) -- (609.46,477.46) -- (480,477.46) -- cycle ; \draw   [line width=1.5]  (660,308.99) -- (609.46,359.53) -- (480,359.53) ; \draw   [line width=1.5]  (609.46,359.53) -- (609.46,477.46) ;
\draw   [line width=1.5]  (400,69.53) -- (450.54,18.99) -- (580,18.99) -- (580,136.92) -- (529.46,187.46) -- (400,187.46) -- cycle ; \draw   [line width=1.5]  (580,18.99) -- (529.46,69.53) -- (400,69.53) ; \draw   [line width=1.5]  (529.46,69.53) -- (529.46,187.46) ;
\draw   [line width=1.5]  (770,69.53) -- (820.54,18.99) -- (950,18.99) -- (950,136.92) -- (899.46,187.46) -- (770,187.46) -- cycle ; \draw   [line width=1.5]  (950,18.99) -- (899.46,69.53) -- (770,69.53) ; \draw   [line width=1.5]  (899.46,69.53) -- (899.46,187.46) ;
\draw   [line width=1.5]  (770,279.53) -- (820.54,228.99) -- (950,228.99) -- (950,346.92) -- (899.46,397.46) -- (770,397.46) -- cycle ; \draw   [line width=1.5]  (950,228.99) -- (899.46,279.53) -- (770,279.53) ; \draw   [line width=1.5]  (899.46,279.53) -- (899.46,397.46) ;
\draw   [dash pattern={on 0.84pt off 2.51pt}]  (290,308.99) -- (529.46,69.53) ;
\draw   [dash pattern={on 0.84pt off 2.51pt}]  (530.54,308.99) -- (770,69.53) ;
\draw   [dash pattern={on 0.84pt off 2.51pt}]  (660,426.92) -- (899.46,187.46) ;
\draw   [dash pattern={on 0.84pt off 2.51pt}]  (660,518.99) -- (899.46,279.53) ;
\draw   [dash pattern={on 0.84pt off 2.51pt}]  (660,390) -- (770,279.53) ;
\draw   [dash pattern={on 0.84pt off 2.51pt}]  (570.54,477.99) -- (530.54,518.99) ;
\draw   [dash pattern={on 0.84pt off 2.51pt}]  (239.46,477.46) -- (480,477.46) ;
\draw   [dash pattern={on 0.84pt off 2.51pt}]  (239.46,570) -- (480,570) ;
\draw   [dash pattern={on 0.84pt off 2.51pt}]  (290,308.99) -- (530.54,308.99) ;
\draw   [dash pattern={on 0.84pt off 2.51pt}]  (529.46,187.46) -- (770,187.46) ;
\draw   [dash pattern={on 0.84pt off 2.51pt}]  (580,136.92) -- (770.54,136.92) ;
\draw   [dash pattern={on 0.84pt off 2.51pt}]  (529.46,69.53) -- (770,69.53) ;
\draw   [dash pattern={on 0.84pt off 2.51pt}]  (290,518.99) -- (530.54,518.99) ;
\draw   [dash pattern={on 0.84pt off 2.51pt}]  (290,426.92) -- (529.46,187.46) ;
\draw   [dash pattern={on 0.84pt off 2.51pt}]  (239.46,570) -- (239.46,477.46) ;
\draw   [dash pattern={on 0.84pt off 2.51pt}]  (290,518.99) -- (290,426.45) ;
\draw   [dash pattern={on 0.84pt off 2.51pt}]  (660,519.46) -- (660,426.92) ;
\draw   [dash pattern={on 0.84pt off 2.51pt}]  (770,279.53) -- (770,186.99) ;
\draw   [dash pattern={on 0.84pt off 2.51pt}]  (899.46,280) -- (899.46,187.46) ;
\draw    (450.54,349.46) -- (282,518) ;
\draw    (450,188) -- (450.54,349.46) ;
\draw    (450.54,349.46) -- (490,350) ;
\draw    (660,350) -- (770,350) ;
\draw   [dash pattern={on 0.84pt off 2.51pt}]  (480,570) -- (480,477.46) ;

\draw  [draw opacity=0][fill={rgb, 255:red, 255; green, 255; blue, 255 }  ,fill opacity=1 ][line width=0.75]   (778,108) -- (892,108) -- (892,148) -- (778,148) -- cycle  ;
\draw (785,116) node [anchor=north west][inner sep=0.75pt]  [font=\Large] [align=left] {$\bcalH^{( j_{1} k_{M_{2}} \ell _{M_{3}})}$};
\draw  [draw opacity=0][fill={rgb, 255:red, 255; green, 255; blue, 255 }  ,fill opacity=1 ][line width=0.75]   (772,318) -- (898,318) -- (898,358) -- (772,358) -- cycle  ;
\draw (785,325) node [anchor=north west][inner sep=0.75pt]  [font=\Large] [align=left] {$ \bcalH^{( j_{M_{1}} k_{M_{2}} \ell _{M_{3}})}$};
\draw  [draw opacity=0][fill={rgb, 255:red, 255; green, 255; blue, 255 }  ,fill opacity=1 ][line width=0.75]   (490,398) -- (592,398) -- (592,438) -- (490,438) -- cycle  ;
\draw (500,408) node [anchor=north west][inner sep=0.75pt]  [font=\Large] [align=left] {$ \bcalH^{( j_{1} k_{M_{2}} \ell _{1})}$};
\draw  [draw opacity=0][fill={rgb, 255:red, 255; green, 255; blue, 255 }  ,fill opacity=1 ][line width=0.75]   (488,608) -- (602,608) -- (602,648) -- (488,648) -- cycle  ;
\draw (500,617) node [anchor=north west][inner sep=0.75pt]  [font=\Large] [align=left] {$ \bcalH^{( j_{M_{1}} k_{M_{2}} \ell _{1})}$};
\draw  [draw opacity=0][fill={rgb, 255:red, 255; green, 255; blue, 255 }  ,fill opacity=1 ][line width=0.75]   (120,608) -- (222,608) -- (222,648) -- (120,648) -- cycle  ;
\draw (131,617) node [anchor=north west][inner sep=0.75pt]  [font=\Large] [align=left] {$ \bcalH^{( j_{M_{1}} k_{1} \ell _{1})}$};
\draw  [draw opacity=0][fill={rgb, 255:red, 255; green, 255; blue, 255 }  ,fill opacity=1 ][line width=0.75]   (128,404) -- (218,404) -- (218,442) -- (128,442) -- cycle  ;
\draw (135,408) node [anchor=north west][inner sep=0.75pt]  [font=\Large] [align=left] {$ \bcalH^{( j_{1} k_{1} \ell _{1})}$};
\draw  [draw opacity=0][fill={rgb, 255:red, 255; green, 255; blue, 255 }  ,fill opacity=1 ][line width=0.75]   (418,112) -- (520,112) -- (520,152) -- (418,152) -- cycle  ;
\draw (421,116) node [anchor=north west][inner sep=0.75pt]  [font=\Large] [align=left] {$ \bcalH^{( j_{1} k_{1} \ell _{M_{3}})}$};

\end{tikzpicture}}
    \caption{Structure of $\bcalY$ as a concatenation of 3D marginals.}
    \label{fig:repartKargas}
\end{figure}
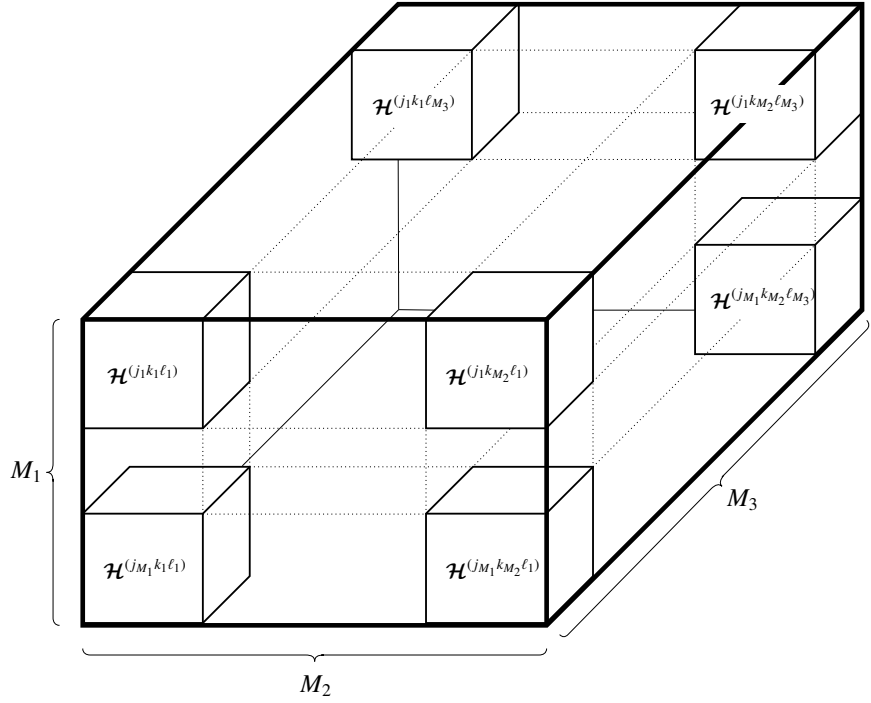

Under such stacking, we can formulate the following obvious lemma.
\begin{lemma}\label{lem:stacking}
The stacked tensor corresponds to the order-$M$ coupled model \eqref{eq:cpd_coupled} if and only if it has 
\begin{equation}
    \label{eq:modelConc}
    \bcalY = \cpdsetp{\lbd;\Bm1,\Bm2,\Bm3},
\end{equation}
with the factor matrices given as
\begin{equation}
    \Bm1 = \begin{bmatrix}
        \Am{j_1} \\ \vdots \\ \Am{M_1}
    \end{bmatrix}\in\dsR^{M_1I\times R}, \quad \Bm2 = \begin{bmatrix}
        \Am{k_1} \\ \vdots \\ \Am{M_2}
    \end{bmatrix}\in\dsR^{M_2I\times R} \quad \Bm3 = \begin{bmatrix}
        \Am{\ell_1} \\ \vdots \\ \Am{M_3}
    \end{bmatrix}\in\dsR^{M_3I\times R}.
    \label{eq:defFactorMatrixConc}
\end{equation}
\end{lemma}

\begin{example}
    \label{ex:m5repart}
    For $M=5$ and a partition $\calM_1 =\{1\}$, $\calM_2 =\{2,3\}$ and $\calM_3 =\{4,5\}$, the set of triplets becomes $\calT = \{\triples124, \triples125, \triples134, \triples135\}$.
    Therefore, the blocks of $\bcalY$ are 
    $$\{\Htrip124, \Htrip125, \Htrip134, \Htrip135 \}.$$
\end{example}

\begin{remark}
    As mentioned previously, \Cref{lem:stacking} was in fact used in \cite{n_kargas_tensors_2018}.

   For \cite[Theorem 1]{n_kargas_tensors_2018}, the partition is the following:
    \begin{equation*}
        \calM_1 = \{1\}, \quad \calM_2 = \{2\} \quad \text{and} \quad \calM_3 = \cpdsetp{3,M},
    \end{equation*} 
    which leads to the following triplets $\calT = \{\triples123, \triples124, \ldots, \triples12M\}$.
    Therefore, $\bcalY$ contains the set of ($M-2$) marginals $\{\Htrip123, \Htrip124, \ldots, \Htrip12M\}$.
    
    For \cite[Theorem 2]{n_kargas_tensors_2018}, a more balanced partition is used:
    \begin{equation*}
        \calM_1 = \cpdsetp{1,\left\lfloor\frac{M}3\right\rfloor}, \quad \calM_2 = \cpdsetp{\left\lfloor\frac{M}3\right\rfloor+1,\left\lfloor\frac{2M}3\right\rfloor} \quad \text{and} \quad \calM_3 = \cpdsetp{\left\lfloor\frac{2M}3\right\rfloor+1,M}.
    \end{equation*} 

    The idea is then to use generic uniqueness results for tensor decomposition of the stacked tensor \eqref{eq:modelConc}.
\end{remark}

\subsection{The importance of constraints and the main result}
The constraints in the original PCTF3D model impose additional constraints on the stacked CPD model \eqref{eq:modelConc}.
First, nonnegativity constraints \eqref{eq:lbdNN} and \eqref{eq:AmNN} require the factors of \eqref{eq:modelConc} to be also non-negative:
\begin{equation*}
    \lbd\geq 0, \quad \Bm1\geq0, \quad \Bm2\geq0, \quad \Bm3\geq0.
\end{equation*}
Second, the sum-to-one constraint \eqref{eq:lbdS21} on $\lbd$ remains ($\vUn^\T_R\lbd = 1$) and the sum-to-one constraints on the factors lead to the more complicated constraints as follows.
\begin{remark}\label{rem:constraints_Bm}
    The factors $\Am{m}$ satisfy the constraints \eqref{eq:AmS21} if and only if $\Bm1, \Bm2, \Bm3$ satisfy
    \begin{equation*}
    (\bfI_{M_1}\kron\vUn^\T_I)\Bm1 = \vUn_{M_1 \times R}, \quad (\bfI_{M_2}\kron\vUn^\T_I)\Bm2 = \vUn_{M_2 \times R}, \quad (\bfI_{M_3}\kron\vUn^\T_I)\Bm3 = \vUn_{M_3 \times R};
    \end{equation*}
    i.e. every column of $\Bm1$, $\Bm2$, $\Bm3$, must satisfy $M_1,M_2,M_3$ independent conditions, respectively.
\end{remark}

\Cref{rem:constraints_Bm} shows that we cannot formally use generic uniqueness results for $IM_1 \times IM_2 \times IM_3$ tensors (as it was done in \cite{n_kargas_tensors_2018}), since the factors are non-generic (belong to a set of measure $0$ due to the constraints).
Though the results in \cite{n_kargas_tensors_2018} remain correct, our main result shows that we can still use generic uniqueness results for tensors of reduced size.

\begin{theorem}
    \label{thm:IdentCoupledModel}
    Let $\calT = \calM_1\times\calM_2\times\calM_3$ be the Cartesian coupling.
    If, for a tensor of size $((I-1)M_1+1)\times ((I-1)M_2+1)\times ((I-1)M_3+1)$, the real rank-$R$ CP model is identifiable, then the PCTF3D model $(\mu,\Tht)$ (see\eqref{eq:pctf3d_model_unconstrained},\eqref{eq:Tht_pctf3d}) of rank $R$ for the Cartesian coupling is identifiable.
\end{theorem}
The proof of the theorem is simple and relies on the properties of polynomial models. We postpone the proof to the next subsection, and discuss first the consequences of this result.

\subsection{Identifiability result for an even partition of variables} \label{sec:identBest}

In this subsection, let $M = 3\left\lfloor \frac{M}3\right\rfloor+\varepsilon$ (with $0\leq \varepsilon< 3$) and the partition of variables is considered as follows:
\begin{itemize}
    \item If $\varepsilon=0$, $M_1=M_2=M_3=\left\lfloor M/3\right\rfloor$,
    \item If $\varepsilon=1$, $M_1=\left\lfloor M/3\right\rfloor+1$ and $M_2=M_3=\left\lfloor M/3\right\rfloor$,
    \item If $\varepsilon=2$, $M_1=M_2=\left\lfloor M/3\right\rfloor+1$ and $M_3=\left\lfloor M/3\right\rfloor$.
\end{itemize}
Note that this partition groups variables in the most balanced way possible in the sense that the resulting 3D tensor (see \Cref{fig:repartKargas}) is the most cubic regarding the repartition of variables.
\begin{theorem}
    Let $\bcalH$ an order-$M$ PMF tensor of size $I\times \cdots \times I$ ($M>3$) and let $\calT$ the Cartesian coupling with an even partition of variables.
    Then, the CPD of $\bcalH$ is identifiable from the 3D marginal tensors $\{\Hjkl\}_{\jkl\in\calT}$ if:
    \begin{equation}
        R\leq \frac{\left\lfloor M/3\right\rfloor^2(I-1)^2}3 -\left\lfloor M/3\right\rfloor(I-1).
        \label{eq:ident3Deven}
    \end{equation}
\end{theorem}
\begin{proof}
    In this proof, we restrain ourselves to the case of $M =3K$, hence $M_1 = M_2 = M_3 = M/3$.
    First, by invoking both Theorem \ref{thm:IdentCoupledModel}, the identifiability of the CP model of $\bcalH$ is equivalent with the identifiability of an order-3 of size $(1+\frac{M}3(I-1))\times (1+\frac{M}3(I-1)) \times (1+\frac{M}3(I-1))$.
    Theorem \ref{eq:thm_ident_3d} may then be applied and gives identifiability of the model for:
    \begin{align*}
        R\;\leq \quad & \frac{\left(1+\left\lfloor M/3\right\rfloor(I-1)\right)^3}{1+M(I-1)} - \left\lfloor M/3\right\rfloor(I-1)-1 \\
        \Longrightarrow \quad R\;\leq \quad & \frac{1+M(I-1)+3\left\lfloor M/3\right\rfloor^2(I-1)^2+\left\lfloor M/3\right\rfloor^3(I-1)^3}{1+M(I-1)} - \left\lfloor M/3\right\rfloor(I-1)-1 \\
        \Longrightarrow \quad R\;\leq \quad & \left\lfloor M/3\right\rfloor^2(I-1)^2\underbrace{\frac{3+\left\lfloor M/3\right\rfloor(I-1)}{1+M(I-1)}}_{\geq\frac13} - \left\lfloor M/3\right\rfloor(I-1),
    \end{align*}
    which finishes the proof for this case.
    For the other two cases, the proof is similar and \eqref{eq:ident3Deven} is easier to prove as it is less optimal than the case of $M$ being a multiple of 3.
\end{proof}
As mentioned in \Cref{sec:cartesianProd}, if any partially coupled case has identifiability, it guarantees the identifiability of the fully coupled model.
Therefore, the uniqueness result of \Cref{eq:ident3Deven} is also true for the fully coupled case.
\Cref{fig:identBound} shows the different uniqueness results obtained in the case of a full coupling -- hence for the FCTF3D method.
This plot shows that the identifiability bounds we can obtain are better than the bound of the original FCTF3D paper \cite{n_kargas_tensors_2018}.
\begin{figure}
    \centerline{\includegraphics[width = \linewidth,clip,trim = {85 5 90 5}]{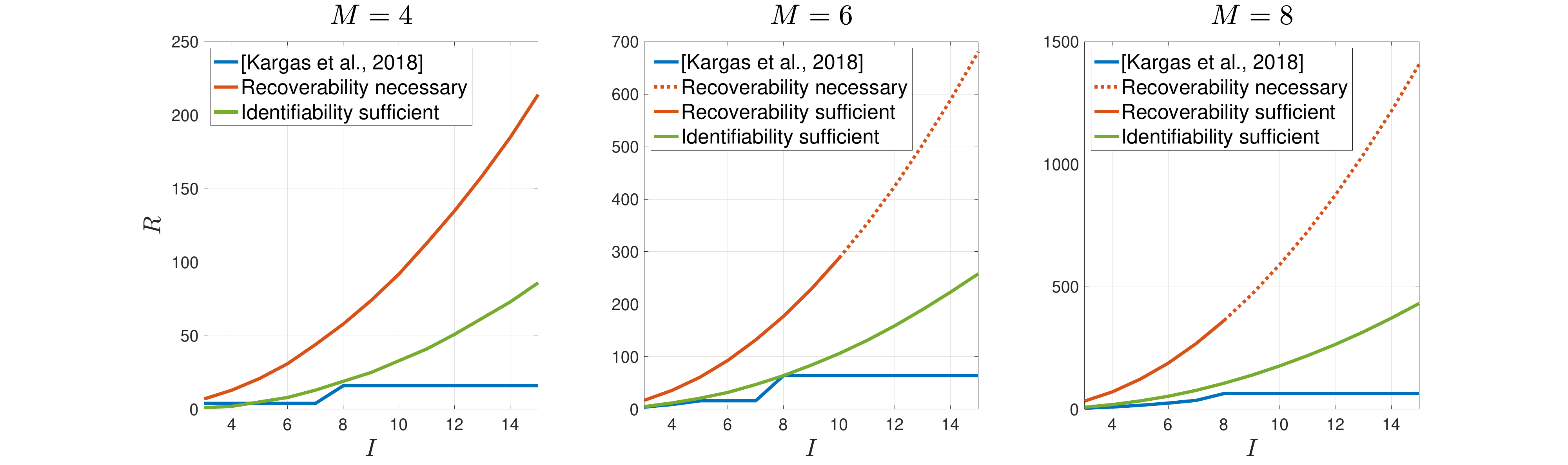}}
    \caption{Identifiability bounds in the fully coupled case along with recoverability results.}
    \label{fig:identBound}
\end{figure}

\subsection{Proof of the main result for Cartesian coupling}

\begin{proof}[Proof of \Cref{thm:IdentCoupledModel}]
Thanks to the results of \Cref{sec:recoAlgo}, we will consider identifiability of the reparameterized model 
$(\mu,\Tht)$ defined by \eqref{eq:defTht}-\eqref{eq:mappinghtau}.
Therefore, we take the truncated $\tht \in \dsR^{R(1+M(I-1))}$ (see \eqref{eq:deftht}) as our parameter vector and 
assume $\uAm{m} = \calP(\uAm{m})$ in \eqref{eq:modelConc} and \eqref{eq:defFactorMatrixConc}.
 
For $\tht$ belonging to the set of constraints, the factors $\Bm1, \Bm2, \Bm3$ in \eqref{eq:modelConc} can be expressed as follows: 
\begin{equation}
    \label{eq:relBQuB}
    \Bm1 = \bfQ_1\bfC^{(1)}, \quad 
    \Bm2 = \bfQ_2\bfC^{(2)} \quad \text{and} \quad 
    \Bm3 = \bfQ_3\bfC^{(3)},
\end{equation}
where, $\bfQ_1$ (respectively $\bfQ_2$ and $\bfQ_3$) are constant matrices of size $M_1I\times (M_1(I-1)+1)$ (respectively $M_2I\times (M_2(I-1)+1)$ and $M_3I\times (M_3(I-1)+1)$), defined below, and
    \begin{equation*}
        \Cm1 = \underbrace{\begin{bNiceArray}{ccc}[margin]
            1 & \Cdots & 1 \\
            \Block{1-3}{\uAm{j_1}} \\
            & \vdots & \\
            \Block{1-3}{\uAm{j_{M_1}}} \\
        \end{bNiceArray}}_{(M_1(I-1)+1)\times R}, \quad
        \Cm2 = \underbrace{\begin{bNiceArray}{ccc}[margin]
            1 & \Cdots & 1 \\
            \Block{1-3}{\uAm{k_1}} \\
            & \vdots & \\
            \Block{1-3}{\uAm{k_{M_2}}} \\
        \end{bNiceArray}}_{(M_2(I-1)+1)\times R}, \quad 
        \Cm3 = \underbrace{\begin{bNiceArray}{ccc}[margin]
            1 & \Cdots & 1 \\
            \Block{1-3}{\uAm{\ell_1}} \\
            & \vdots & \\
            \Block{1-3}{\uAm{\ell_{M_3}}} \\
        \end{bNiceArray}}_{(M_3(I-1)+1)\times R},
    \end{equation*}
For example, the matrix $\bfQ_1$ has the form shown in \Cref{fig:matrQ}.

    \begin{figure}[bt]
        \centering
\begin{equation*}
    \bfQ_1 = \begin{bNiceArray}{c|cccc|cccc|cccc|cccc}[margin]
            0      & \Block{4-4}{\jac{\calP}} & & & & \Block{4-4}{\bm{0}_{I\times I-1}} & & & & & & & & \Block{4-4}{\bm{0}_{I\times I-1}} & & & \\
            \vdots & & & & & & & & & & \Hdotsfor{2} & & & & & \\
            0      & & & & & & & & & & & & & & & & \\
            1      & & & & & & & & & & & & & & & & \\ 
            \hline
            0      & \Block{4-4}{\bm{0}_{I\times I-1}} & & & & \Block{4-4}{\jac{\calP}} & & & & \phantom{0} & & & & & & \phantom{\vdots} & \\
            \vdots & & & & & & & & & & \Ddots & & & & & \Vdots & \\
            0      & & & & & & & & & & & & & & & & \\
            1      & & & & & & & & & & & & \phantom{0} & & & \phantom{\vdots} & \\ 
            \hline
            0      & & \phantom{\vdots} & & & \phantom{0} & & & & \phantom{0} & & & & \Block{4-4}{\bm{0}_{I\times I-1}} & & & \\
            \vdots & & \Vdots & & & & \Ddots & & & & \Ddots & & & & & & \\
            0      & & & & & & & & & & & & & & & & \\
            1      & & \phantom{\vdots} & & & & & & \phantom{0} & & & & \phantom{0} & & & & \\ 
            \hline
            0      & \Block{4-4}{\bm{0}_{I\times I-1}} & & & & & & & & \Block{4-4}{\bm{0}_{I\times I-1}} & & & & \Block{4-4}{\jac{\calP}} & & & \\
            \vdots & & & & & & \Hdotsfor{2} & & & & & & & & & \\
            0      & & & & & & & & & & & & & & & & \\
            1      & & & & & & & & & & & & & & & & 
        \end{bNiceArray},
\end{equation*}
        \caption{The matrix $\bfQ_1$.
        The matrix $\jac{\calP}\in\dsR^{I\times (I-1)}$ is the Jacobian matrix of the mapping for truncated factors (see \Cref{eq:projS21}).}
        \label{fig:matrQ}
    \end{figure}

Then, by multilinearity, the tensor $\bcalY$ can be expressed in terms of a smaller tensor:
\[
\vectorize \bcalY (\tht) = \bfQ_3\kron\bfQ_2\kron\bfQ_1 \mu_2(\tht),
\]
where
 \begin{equation}
        \mu_2(\tht) = \vectorize(\bcalZ) \quad \text{where } \bcalZ = \cpdsetp{\lbd; \Cm1, \Cm2, \Cm3}.
        \label{eq:CPcone1}
    \end{equation}
The model $(\mu_2,\Tht)$ is thus a re-parametrization of the original model, and, by and \Cref{rem:reparameterization_PCTF3D} PCTF3D is identifiable if and only if $(\mu_2,\Tht)$ is identifiable.

Finally, we remark that model $(\mu_2,\Tht)$ is a submodel of the polynomial additive model in \Cref{lem:reduced_parameterization} for $I_1 = (M_1(I-1)+1)$, $I_2 = (M_2(I-1)+1)$ and $I_3 = (M_3(I-1)+1)$.
Moreover, the constraint set \eqref{eq:defTht} in this case is of positive Lebesgue measure,
hence identifiability of $I_1 \times I_2 \times I_3$ CPD implies identifiability of the Cartesian coupling by \Cref{lem:identifiable-submodel}. 
\end{proof}

\section{Conclusion}

In this article, we proposed a new model for PMF estimation which is a fundamental issue in signal processing.
Our new model breaks the curse of dimensionality at a newer level.
Indeed, by coupling only a subset of 3D marginals, we are able to reduce and control the complexity of our approach.
With this new model, we proposed to choose the marginals with the framework of hypergraphs.
Hypergraphs permit to introduced examples of coupling strategies such as random couplings and balanced couplings which ensures that variables are represented evenly in terms of occurrences in the graph.

This article, as the second part of a two-parts article (see \cite{pctf3d_part1}), mainly focus on the question of uniqueness of our newly-proposed model. 
In the sense of recoverability, the model is studied with a new algorithm that returns the maximal recoverable rank and giving a certificate of recoverability in a particular setting.
When applied to random couplings, this algorithm revealed the existence of so-called defective cases which have been explained with the structure of the Jacobian matrix.
Because no defective cases where seen for balanced couplings, it shows their utility in practical cases.

Concerning identifiability, it has been studied in the case of Cartesian couplings, a strategy that partitions variables into 3 groups.
By studying the order-3 tensor model obtained with these couplings, we have been able to prove a sufficient condition on identifiability, by studying the algebraic geometry properties of the sum-to-one constraints of the model.
This new identifiability results reveals better results than the literature (and especially \cite{n_kargas_tensors_2018}) without relaxing the sum-to-one constraints on the factors.

\appendix
\section{Technical proofs}\label{sec:tech_proofs}

\begin{proof}[Proof of \Cref{lem:reduced_parameterization}]
Suppose, by contradiction, that \eqref{eq:rank1_parameterization_ones} is not identifiable.
Then there exists a Euclidean ball of $\mathcal{B} \in \dsR^{R(I+J+K-2)}$ (of positive Lebesgue measure) of the parameters
such that for all $(\lambda_{r},\underline{\bfa}_r, \underline{\bfb}_r, \underline{\bfc}_r)_{k=1}^R \in \mathcal{B}$ the additive decomposition for $\mu_1$ as in \eqref{eq:rank1_parameterization_ones} 
\[
\vectorize(\bcalX) = \mu_1(\lambda_{1},\underline{\bfa}_1, \underline{\bfb}_1, \underline{\bfc}_1) + \cdots + 
 \mu_1(\lambda_{R},\underline{\bfa}_R, \underline{\bfb}_R, \underline{\bfc}_R)
\]
is non-unique.

But then in this case, the CPD of underlying tensor
\[
\bcalX = \sum\limits_{r=1}^R \lambda_r {\bfa}_r \out{\bfb}_r\out{\bfc}_r
\]
is non-unique for any $\lambda_{r}, \bfa_r = \alpha_r \left[\begin{smallmatrix}\underline{\bfa}_r \\1 \end{smallmatrix}\right], \bfb_r = \beta_r \left[\begin{smallmatrix}\underline{\bfb}_r \\1 \end{smallmatrix}\right], \bfc_r = \gamma_r \left[\begin{smallmatrix}\underline{\bfc}_r \\1 \end{smallmatrix}\right]$, with
\[
(\lambda_{r},\underline{\bfa}_r, \underline{\bfb}_r, \underline{\bfc}_r)_{k=1}^R \in \mathcal{B}, \alpha_r, \beta_r, \gamma_r \neq 0,
\]
which defines the set of parameters $(\lambda_{r},\underline{\bfa}_r, \underline{\bfb}_r, \underline{\bfc}_r)_{k=1}^R \in \dsR^{R(I+J+K+1)}$ of positive measure, thus we arrive at the contradiction.
\end{proof}

\begin{proof}[Proof of \Cref{prop:relaxing_sum_to_one}]
Define by $\Tht^{(\alpha)} \subset \Tht'$, the set of all parameters $\tht \in \Tht'$ such that $\left(\sum\limits_{r=1}^R \lambda_r\right) = \alpha > 0$ (in particular, we have $\Tht = \Tht^{(1)}$).
Note that from sum-to-one constraints on the factors, we get for all $\tht \in \Tht^{(\alpha)}$
\[
\alpha = \left(\left(\left(\Hjkl \bullet_3 \vUn^\T_I\right)\bullet_2\vUn^\T_I\right)\bullet_1\vUn^\T_I\right).
\]
for all $\jkl$ and $\Hjkl$ obtained in \eqref{eq:pctf3d_model_unconstrained}.
Therefore, we have $\mu^{-1}(\tht) \subset \Tht^{(\alpha)}$ if $\tht \in \Tht^{(\alpha)}$ (all equivalent parameters must be from the same class).

Let $\rho(\tht)$ be a normalization of the parameter (that replaces all $\lambda_r$ by $\frac{\lambda_r}{\sum\limits_{r=1}^R \lambda_r}$). Then for any $\alpha> 0$, the map $\rho$ is one-to-one from $\Tht^{(\alpha)}$ to $\Tht^{(1)}$.
Therefore, $\mu(\tht)$ has unique (finite number of) decomposition(s) if $\mu(\rho(\tht))$ has unique (finite number of) decomposition(s).
 
Finally, we note that $\Tht' \setminus \cup_{\alpha > 0} \Tht^{(\alpha)} $ is a subset of measure zero, which completes the proof.
\end{proof}

\bibliographystyle{elsarticle-num}

\end{document}